\documentclass[12pt,journal,draftcls,letterpaper,onecolumn]{IEEEtran}
\usepackage[utf8]{inputenc}
\usepackage{verbatim}
\usepackage[cmex10]{amsmath}
\usepackage{amssymb}
\usepackage{mathrsfs}
\usepackage[final]{graphicx}
\usepackage[numbers]{natbib}
\newtheorem{definition}{Definition}

\newtheorem{remark}{Remark}

\newcommand{\expectation}{\ensuremath{\mathbb{E}}}
\newcommand{\Expt}{\expectation}
\newcommand{\probability}{\ensuremath{\mathbb{P}}}
\newcommand{\Prob}{\probability}
\usepackage{url}

\begin{document}

\title{Extremality for Gallager's Reliability Function $E_0$}
\author{\IEEEauthorblockN{Mine Alsan}\\
\small\IEEEauthorblockA{Information Theory Laboratory\\
Ecole Polytechnique F\' ed\' erale de Lausanne\\
CH-1015 Lausanne, Switzerland\\
Email: mine.alsan@epfl.ch}
\normalsize}
\maketitle
\pagestyle{plain}
\thispagestyle{plain}

\maketitle

\let\thefootnote\relax\footnotetext{Part of the material in this paper was presented in part at the IEEE International Symposium on Information Theory, Boston, USA,
July 2012.}

\begin{abstract}
We describe certain extremalities for Gallager's $E_0$ function evaluated under the uniform input distribution for binary input discrete memoryless channels. The results characterize the extremality of the $E_0(\rho)$ curves of the binary erasure channel and the binary symmetric channel among all the $E_0(\rho)$ curves that can be generated by the class of binary discrete memoryless channels whose $E_0(\rho)$ curves pass through a given point $(\rho_0, e_0)$, for some $\rho_0 > -1$.
\end{abstract}




%

\section{Introduction}
While the capacity of a memoryless channel $W$ gives the largest
rate for which reliable communication is possible, the reliability function
$E(R,W)$ provides a finer measure on the quality of the channel: for any $R$ less than the channel capacity, it is possible to find a sequence of
codes of increasing blocklength, each of which of rate at least $R$, and
whose block error probability decays exponentially to zero as the
blocklength increases ---  $E(R,W)$ is the largest possible rate of this decay.

Gallager's classical treatise~\cite{578869} gives a lower bound to $E(R,W)$,
the random coding exponent $E_r(R,W)$ in the form $E_r(R,W)=\max_{\rho\in[0,1]}
E_0(\rho,W)-\rho R$.  Remarkably, this lower bound is tight for rates above
the critical rate $E_0'(1,W)$.  The function $E_0(\rho,W)$ that appears
as an auxiliary function on the road to deriving $E_r(R,W)$ turns out to
be of independent interest in its own right.  In particular, $E_0(\rho,W)/\rho$
is the largest rate for which a sequential decoder can operate while keeping
the $\rho$-th moment of the decoder's computation effort per symbol bounded \cite{481781}.

Previously, we investigated in \cite{6284065} the extremal properties of $E_0(\rho,W)$, for $\rho\in[0, 1]$, for the class of binary input discrete memoryless channels (B-DMC) when the function is evaluated under the uniform input distribution. We have shown that among all channels with a given value of $E_0(\rho_1,W)$, the binary erasure channel (BEC) and the binary symmetric channel (BSC) distinguish themselves in certain ways: they have, respectively, the largest and the smallest value of $E_0'(\rho_2,W)$ 
for any $\rho_1, \rho_2\in[0, 1]$ such that $\rho_2\geq\rho_1$. 
As the random coding exponent is obtained by tracing the map $\rho\to(E_0'(\rho), E_0(\rho)-\rho E_0'(\rho))$, among the simple corollaries of this is the conclusion that of all the symmetric channels with the
same capacity, the BEC and the BSC have the largest and the smallest value of $E_r(R, W)$, a result reported in \cite{6034105}. 

In this paper, we extend the previous extremality analysis of \cite{6284065} to
both the cases when $\rho > 1$ and when $\rho\in(-1, 0)$. The extremal results for $E_0$ in these regions 
are motivated by various error exponents such as the list decoding exponent~\cite{578869}, defined for $\rho > 0$, 
and the exponent which appears in Arimoto's lower bound for the strong converse of the coding theorem \cite{1055007}, defined for $\rho\in(-1, 0]$. 
For a concise list of the definitions of various error exponents involving the $E_0$ function, we refer to \cite{6377299}, 
a recent study which also examined the extremality of $E_0(\rho)$ for $\rho>-1$, but only for the special class of symmetric B-DMCs of the same capacity. 

The results of this paper characterize the extremality of the $E_0$ curves of the BEC and the BSC among all
the $E_0$ curves that can be generated by the class of B-DMCs whose $E_0$ curves pass through a given point $(\rho_0, e_0)$ for some $\rho_0 > -1$. 
We prove that when $\rho_0\in(-1,1]$, these two channels remain extremal along the $E_0(\rho)$ curves for any $\rho > -1$.
We also prove that when $\rho_0 > 1$, while these two channels are extremal along the $E_0(\rho)$ curves for any $\rho\in(-1, 1]$, no extremality beyond $\rho > 1$ can be formulated in general. Furthermore, we show that the conclusion we have mentioned above for $E_0'$ is still valid when $\rho_1\in(0, 1]$ and $\rho_2 \geq \rho_1$ (even for $\rho_2 > 1$), and also when $\rho_1\in(-1, 0]$ and $\rho_2 \leq \rho_1$. Using these, we recover the result of \cite{6377299} which shows that, for any $\rho > -1$, the BEC and the BSC are $E_0$ extremal among the $E_0(\rho)$ curves of all symmetric channels with the same capacity.

The rest of this paper is organized as follows. Section \ref{sec:pre} starts by giving the preliminary definitions, 
and then later derives some basic properties of the $E_0$ curves of BECs and BSCs. 
Subsequently, in Section \ref{sec:ER}, the main results of this paper are stated in Theorem \ref{thm:Extremality_Results}. 
The section follows by some convexity lemmas, the proof of the theorem, and a graphical interpretation of the extremality results. Finally, the last section gives the conclusions.

\section{Preliminaries}\label{sec:pre}
\subsection{Definition of the Random Coding Exponent and $E_0$}
\begin{definition}~\cite[Section 5.6]{578869} 
Given a discrete memoryless channel (DMC) $W$ with input alphabet $\mathcal{X}$
and output alphabet $\mathcal{Y}$, fix a distribution $Q$ on its input
alphabet. Consider the function $E_r(R, Q, W)$ defined as
\begin{equation}\label{eq::Er_R_Q_W}
E_r(R, Q, W) = \displaystyle\max_{\rho\in[0, 1]} \{ E_0(\rho, Q, W) - \rho R \},
\end{equation} 
for $R\geq0$, where
\begin{equation} \label{eq:EoQ}  
 E_0(\rho, Q, W) = -\log \displaystyle\sum_{y\in \mathcal{Y}} \left[ \displaystyle\sum_{x\in \mathcal{X}} Q(x) W(y\mid x)^{\frac{1}{1+\rho}}\right]^{1+\rho},
\end{equation}
with the $\log$ denoting the natural logarithm to the base e.
The random coding exponent of the channel is defined as
\begin{equation} \label{eq:Er}
 E_r(R,W) = \displaystyle\max_{Q}  E_r(R, Q, W).  
\end{equation}  
\end{definition}

Throughout this paper, we fix $\mathcal{X}$ to $\{0, 1\}$ and $Q$ to the uniform input distribution. Then, the expression in \eqref{eq:EoQ} becomes 
\begin{equation} \label{eq:E0} 
E_0(\rho, W) = -\log \displaystyle\sum_{y\in \mathcal{Y}} \left[ \frac{1}{2} W(y\mid 0)^{\frac{1}{1+\rho}} + \frac{1}{2} W(y\mid 1)^{\frac{1}{1+\rho}} \right]^{1+\rho}.   
\end{equation} 
For symmetric channels, the uniform input distribution corresponds to the distribution which maximizes \eqref{eq:Er}~\cite{578869}. The random coding exponent of symmetric channels is then given by
\begin{equation}\label{eq::Er_R_W}
E_r(R, W) = \displaystyle\max_{\rho\in[0, 1]} \{ E_0(\rho, W) - \rho R \}.
\end{equation}
Moreover, the right hand side of \eqref{eq::Er_R_W} gives a lower bound to the random coding exponents of B-DMCs which are not symmetric. 

The properties of $E_0(\rho, W)$ with respect to the variable $\rho$ are summarized in  \cite[Theorem 5.6.3]{578869}. For $\rho \geq 0$, $E_0(\rho, W)$ 
is a positive, concave increasing function in $\rho$. By convexity, the maximization in the right hand side of \eqref{eq::Er_R_W} over $\rho\in[0, 1]$ can be described in terms of the following parametric equations:
\begin{align} \label{eq:parametric}  
 &R(\rho, W) = \frac{\partial}{\partial \rho} E_0(\rho, W), \\
 &E_r(\rho, W) =  E_0(\rho, W) - \rho \frac{\partial}{\partial \rho} E_0(\rho, W),  
\end{align} 
for $R$ in the range 
\begin{equation}
\displaystyle\frac{\partial E_0(\rho, W)}{\partial \rho}\Bigr\rvert_{\rho=1} \leq R \leq \displaystyle\frac{\partial E_0(\rho, W)}{\partial \rho}\Bigr\rvert_{\rho=0}.
\end{equation}

It is shown in \cite[see Figure 5.6.2]{578869} that the symmetric capacity of the channel, 
\begin{equation}
I(W) = \displaystyle\sum_{y\in\mathcal{Y}}\sum_{x\in\{0, 1\}} \frac{1}{2}W(y|x) \log{\displaystyle\frac{W(y|x)}{\frac{1}{2}W(y|0) + \frac{1}{2}W(y|1)}},
\end{equation}
is the slope of the $E_0$ curve at $\rho=0$, i.e.,  
\begin{equation}\label{eq::I_def} 
I(W) = \frac{\displaystyle \partial}{ \displaystyle \partial \rho} E_0(\rho, W)\Bigl\lvert_{\rho=0}.
\end{equation} 

Finally, another channel parameter of interest for DMCs, the cut-off rate, can also be derived from $E_0$, see \cite{10.1109/TIT.2005.862081} for more information on the significance of this parameter. The cut-off rate of a B-DMC when evaluated under the uniform input distribution is given by $E_0(1, W)$. 

\subsection{Description of $E_0$ by R\'{e}nyi's Entropy Functions}
In this section, we mention an alternative description of $E_0(\rho, W)/\rho$, which also appears in \cite{370121} and \cite{481781}, using the concept of R\'{e}nyi's entropy functions. This gives an interpretation to $E_0(\rho, W)/\rho$ as a general measure of information.

R\'{e}nyi's entropy function of order $\alpha$ of a discrete random variable $X \sim P(x)$ is defined in~\cite{6016020} as
\begin{equation}
H_{\alpha}(X) = \frac{\alpha}{1-\alpha} \log{\left( \displaystyle\sum_{x} P(x)^{\alpha}\right) ^{\frac{1}{\alpha}}}.
\end{equation}

This definition is extended to the R\'{e}nyi's conditional entropy function of order $\alpha$ of a discrete random variable $X$ given $Y$ with joint distribution $P(x, y)$ in~\cite{article02} as 
\begin{align}
\label{eq:Renyi_cond}H_{\alpha}(X \mid Y) &= \frac{\alpha}{1-\alpha} \log{\displaystyle\sum_{y}\left( \displaystyle\sum_{x} P(x,y)^{\alpha}\right) ^{\frac{1}{\alpha}}}  \\
&= H_{\alpha}(X) + \frac{\alpha}{1-\alpha} \log{\displaystyle\sum_{y}\left( \displaystyle\sum_{x} Q(x) P(y\mid x)^{\alpha}\right) ^{\frac{1}{\alpha}}}, 
\end{align}
where $Q(x) = \frac{\displaystyle P(x)^{\alpha}}{\displaystyle\sum_{x} P(x)^{\alpha}}$ is `tilted' probability distribution. Although different definitions are proposed in the literature for a possible extension of R\'{e}nyi's entropy function to a quantity similar to the conditional entropy function, as one suitable for this study, we use 
the definition in \eqref{eq:Renyi_cond}.

Taking a uniform input distribution and letting $\alpha = \frac{\displaystyle1}{\displaystyle1+\rho}$, we get
\begin{equation}\label{eq:Renyi}
H_{\frac{1}{1+\rho}}(X) = \frac{1}{\rho} \log{\left( \displaystyle\sum_{x} P(x)^{\frac{1}{1+\rho}}\right) ^{\frac{1}{1+\rho}}},  
\end{equation}
\begin{equation}\label{eq:CondiRenyi}
H_{\frac{1}{1+\rho}}(X \mid Y) = H_{\frac{1}{1+\rho}}(X) + \frac{1}{\rho} \log{\displaystyle\sum_{y}\left( \displaystyle\sum_{x} P(x) P(y\mid x)^{\frac{1}{1+\rho}}\right) ^{1+\rho}}.
\end{equation}
Hence, from the definition of $E_0(\rho, W)$ in \eqref{eq:E0}, we deduce
\begin{equation}\label{eq:Eo_over_ro}
 \frac{E_0(\rho, W)}{\rho} = H_{\frac{1}{1+\rho}}(X) - H_{\frac{1}{1+\rho}}(X \mid Y). 
\end{equation}
The quantity in the right hand side of \eqref{eq:Eo_over_ro} is called as the mutual information of order $\frac{\displaystyle1}{\displaystyle1+\rho}$
in~\cite{article02}. Moreover, the following properties are proved: 
\begin{itemize}
 \item $\displaystyle\lim_{\alpha \to 1} H_{\alpha}(X) = H(X)$,
 \item $H_{\alpha}(X \mid Y) \leq H_{\alpha}(X)$, i.e ``conditioning reduces entropy'' is valid for R\'{e}nyi's entropy function, as it is in the Shannon entropy case,
 \item $\frac{\displaystyle E_0(\rho, W)}{\displaystyle\rho}$ is a decreasing function in $\rho$ with $\lim_{\rho\to 0}\frac{\displaystyle E_0(\rho,W)}{\displaystyle \rho} = I(W)$.
\end{itemize}

\subsection{An Alternative Representation of $E_0$ for B-DMCs}
The extremality results we will prove in Section \ref{sec:ER} will be based neither on the `raw definition' of $E_0(\rho, W)$ in \eqref{eq:E0}, nor on the interpretation in terms of Renyi's entropy functions of \eqref{eq:Eo_over_ro}. Instead, we will make use of a description of $E_0(\rho, W)$ introduced by~\cite{notes1} which is more suitable for deriving extremal bounds.  

For a given symmetric B-DMC $W:\mathcal{X}\to\mathcal{Y}$ and a fixed $\rho > -1$, \cite{notes1} shows that there exists a random variable $Z$ taking values in the $[0, 1]$ interval such that 
\begin{equation}\label{eq:Eo}
 E_0( \rho, W) = -\log{\Expt\left[g(\rho, Z)\right]},   
\end{equation}
where the function $g(\rho, z)$ is defined as
\begin{equation}\label{eq:g}
 g(\rho, z) =  \left(\frac{1}{2}\left( 1 + z\right)^{\frac{1}{1+\rho}} + \frac{1}{2}\left(1 - z\right)^{\frac{1}{1+\rho}} \right)^{1+\rho},
\end{equation}
for $\rho\in\mathbf{R}\setminus\{-1\}$ and $z\in[-1, 1]$. 
To see this, define
\begin{equation}\label{eq:dist}
 W(y)=\frac{W(y\mid0) + W(y\mid1)}{2},
\end{equation}
and
\begin{equation}\label{eq:dist}
\Delta(y) = \frac{W(y\mid0)-W(y\mid1)}{W(y\mid0)+W(y\mid1)},
\end{equation}
for $y\in\mathcal{Y}$, so that  $W(y\mid0)= W(y)\left(1+\Delta(y)\right)$ and $W(y\mid1) = W(y)\left(1-\Delta(y)\right)$. Then, one can manipulate \eqref{eq:E0} to find that $Z = \lvert\Delta(Y)\rvert$ with $Y\sim W(y)$ in \eqref{eq:Eo}.

The next lemma gives the first and the second order properties of $g(\rho, z)$ with respect to the variable $z$. The proof is carried in Appendix I.

\newtheorem{lemma}{Lemma}
\begin{lemma}\label{lem:g_concavity}
The function $g(\rho, z)$ defined in \eqref{eq:g} is a concave non-increasing function in $z\in[0, 1]$ for $\rho\in(-\infty, -1)\cup[0, \infty)$, and a convex non-decreasing function in $z\in[0, 1]$ for $\rho\in(-1, 0]$. As $g(\rho, z)$ is symmetric around $z=0$, these properties also determine the function's behavior for $z\in[-1, 0]$. 
\end{lemma}

We denote by $g^{-1}(\rho, t)$ the inverse of the function $g(\rho, z)$ with respect to its second argument. The variable $t$ always takes values from a subset of the interval $[0, 2]$. More specifically, $t\in[2^{-\rho}, 1]$ when $\rho\geq 0$, and $t\in[1, 2^{-\rho}]$ when $\rho\in(-1, 0)$. For shorthand notation, we denote the range of possible values by $t\in[2^{-\rho}, 1]\cup[1, 2^{-\rho}]$, for $\rho>-1$.

Finally, we note that by using \eqref{eq:Eo}, the function $R(\rho, W) = \displaystyle\frac{\partial}{\partial\rho}E_0( \rho, W)$ can be written as 
\begin{equation}\label{eq::R_canonical}
 R(\rho, W) = \displaystyle\frac{-\partial\Expt\left[ g(\rho, Z)\right]/\partial\rho}{\Expt\left[g(\rho, Z)\right]}
 = \displaystyle\frac{\Expt\left[-\partial g(\rho, Z)/\partial\rho\right]}{\Expt\left[g(\rho, Z)\right]},
\end{equation}
where the second equality follows by the dominated convergence theorem.

\subsection{Fun facts about $E_0$ and $E_0'$ of BECs and BSCs}
In this section, we explain some simple facts related to the $E_0$ curves of BECs and BSCs. We will be using some of these facts many times throughout the results section. 

Consider first the representation in \eqref{eq:Eo}. It is not difficult to see that the BECs and the BSCs are special cases of this representation. 

\newtheorem{fact}{Fact}
\begin{fact}\label{fact:Z_bec_bsc}\cite{notes1}
The random variable $Z_{BEC}$ of a BEC is $\{0, 1\}$ valued and satisfy $\Prob[Z_{BEC} = 0] = \epsilon$, where $\epsilon\in[0, 1]$ is the erasure probability of the channel. The random variable $Z_{BSC}$ of a BSC is a constant given by $z_{BSC} = 1-2x$ assuming that $x\in[0, 0.5]$ is the crossover probability of the channel.
\end{fact}

It is well known that the set of BECs and BSCs are ordered in terms of their channel capacities: if the chances of an erasure to happen at the output of a BEC model, or similarly of a bit flip at the output of a BSC model is increasing, the transmission capacities shall decrease, see for instance the textbook \cite{578869}. Intuitively, we expect this graceful degradation to order as well other measures of channel quality. For that purpose, we start by computing the $E_0$ and $E_0'$ parameters of a BEC and a BSC as a function of the erasure probability and the crossover probability of the channels. Let $BEC$ be a BEC with erasure probability $\epsilon\in[0, 1]$. Then, one can easily derive that
\begin{equation}\label{eq:E0_bec}
E_0(\rho, BEC) = -\log{(2^{-\rho}(1-\epsilon) + \epsilon)},
\end{equation}
and
\begin{equation}\label{eq:R_bec}
R(\rho, BEC) = \displaystyle\frac{\partial}{\partial\rho}E_{0}(\rho, BEC) = 
\displaystyle\frac{2^{-\rho} (1 - \epsilon) \log{2}}{2^{-\rho} (1 - \epsilon) + \epsilon}.
\end{equation}
Let $BSC$ be a BSC with crossover probability $x\in[0, 0.5]$. In this case, we are saved from the trouble by  \cite[Example 1 p.146]{578869} which has the derivation of the $E_0$ parameter of a BSC in Equation (5.6.40) and its rate parameter in Equation (5.6.41). Rewriting these equations, we get 
\begin{equation}\label{eq:E0_bsc}
E_0(\rho, BSC) = \rho - (1+\rho)\log\left(x^{\frac{1}{1+\rho}} + (1-x)^{\frac{1}{1+\rho}}\right),
\end{equation} 
and
\begin{equation}\label{eq:R_bsc}
R(\rho, BSC) = 1-\mathcal{H}(\delta),
\end{equation}
where $\delta = \displaystyle\frac{x^{\frac{1}{1+\rho}}}{x^{\frac{1}{1+\rho}}+(1-x)^{\frac{1}{1+\rho}}}$. 

Now, we show that these parameters are monotone functions in the erasure/crossover probabilities of the channels. 
\begin{lemma}\label{lem:ordering_bec_bsc}
For any $\rho\geq 0$, $E_0(\rho, BEC)$ $\left(E_0(\rho, BSC)\right)$ is
decreasing in $\epsilon$ $(x)$. For any $\rho\in(-1, 0]$, $E_0(\rho, BEC)$ $(E_0(\rho, BSC))$ is increasing in $\epsilon$ $(x)$. Moreover, for any $\rho>-1$, $R(\rho, BEC)$ $(R(\rho, BSC))$ is decreasing in $\epsilon$ $(x)$.  
\end{lemma}
\begin{proof}

Taking the first derivative of \eqref{eq:E0_bec} with respect to $\epsilon$, we get
\begin{equation}
\displaystyle\frac{\partial}{\partial\epsilon}E_{0}(\rho, BEC) =
-\displaystyle\frac{1-2^{-\rho}}{2^{-\rho}(1-\epsilon)+\epsilon}.
\end{equation}
One can check that 
\begin{equation}
\displaystyle\frac{\partial}{\partial\epsilon}E_{0}(\rho, BEC) \begin{cases}
> 0, & \hbox{for } \rho\in(-1, 0)\\
= 0, & \hbox{for } \rho = 0 \\
< 0 & \hbox{for }  \rho > 0
\end{cases}.
\end{equation}
As $E_{0}(0, W) = 0$, the $E_0$ curves of all BECs will be ordered such that while for $\rho > 0$ the $E_0$ curves of BECs with smaller erasure probabilities will be larger, for $\rho\in(-1,0)$ the opposite will be true.
 
Now, we show an ordering also holds for the $R$ parameters of BECs.
Taking the first derivative of \eqref{eq:R_bec} with respect to $\epsilon$, we get 
\begin{equation}
\displaystyle\frac{\partial}{\partial\epsilon}R(\rho, BEC) =
-\displaystyle\frac{2^\rho \log{2}}{(1 + (-1 + 2^\rho) \epsilon)^2} < 0.
\end{equation}
Hence, the rate parameters will be decreasing with the erasure probability of the channel for any $\rho>-1$. This completes the proof for the BEC.

Now, we prove the claims for the set of BSCs. First, we note that the term inside the logarithm in  \eqref{eq:E0_bsc} satisfies for $x\in[0, 0.5]$
\begin{equation}
\displaystyle\frac{\partial}{\partial x}\left(x^{\frac{1}{1+\rho}} + (1-x)^{\frac{1}{1+\rho}}\right)
=  \displaystyle\frac{x^{-\frac{\rho}{1+\rho}} - (1-x)^{-\frac{\rho}{1+\rho}}}{1+\rho}\\
= \begin{cases}
< 0, & \hbox{for } \rho\in(-1, 0)\\
= 1, & \hbox{for } \rho = 0 \\
> 0 & \hbox{for }  \rho > 0
\end{cases}.
\end{equation}
Hence, we also have
\begin{equation}
\displaystyle\frac{\partial}{\partial x}E_{0}(\rho, BSC) \begin{cases}
> 0, & \hbox{for } \rho\in(-1, 0)\\
= 0, & \hbox{for } \rho = 0 \\
< 0 & \hbox{for }  \rho > 0
\end{cases},
\end{equation}
which proves the claimed ordering for $E_{0}(\rho, BSC)$. To prove the claim for $R(\rho, BSC)$, we simply note that in \eqref{eq:R_bsc}, for $x\in[0, 0.5]$, we have $\delta\in[0, 0.5]$ increasing in $x$ and the binary entropy function $\mathcal{H}(\delta)$ increasing in $\delta\in[0, 0.5]$. As a result, 
\begin{equation}
\displaystyle\frac{\partial}{\partial x}R(\rho, BSC) < 0,
\end{equation}
as claimed.
\end{proof}

By this lemma, the second fact is in order:

\begin{fact}\label{fact:ordering_bec_bsc}
For any $\rho >-1$, the class of BECs and the class of BSCs ($x\in[0,0.5]$) are strictly ordered in their $E_0(\rho, W)$ parameters, except at $\rho=0$ where $E_0(0, W) = 0$, and in their $R(\rho, W)$ parameters.   
\end{fact}

The ordering we have just discussed is not peculiar to BECs and BSCs and can be generalized to more general classes of channels such as degraded ones. However, Lemma \ref{lem:ordering_bec_bsc} will be sufficient for our purpose as the derivations of Section \ref{sec:ER} does not need results of such a generality.

Next, we argue the validity of an assumption we will encounter in the hypothesis of the main theorem.  
\begin{lemma}\label{lem::cando}
For any given B-DMC $W$ and any fixed $\rho>-1$, there exist a BEC $BEC$ and a BSC $BSC$ such that
\begin{equation}\label{eq:same_E0}
E_0(\rho, W) = E_0(\rho, BEC) = E_0(\rho, BSC).
\end{equation}
The erasure probability of $BEC$ and the crossover probability
of $BSC$ depend both on the channel $W$ and the parameter $\rho$.
\end{lemma}
\begin{proof}
Observe that, by \eqref{eq:Eo}, the equality of the $E_0$ functions in \eqref{eq:same_E0} is equivalent to the equality of 
\begin{equation}\label{eq:eq_Z}
 \Expt\left[g(\rho, Z)\right] = \Expt\left[g(\rho, Z_{BEC})\right] = g(\rho, z_{BSC}),
\end{equation}
where  $Z$, $Z_{\mathsf{BEC}}$ and $z_{\mathsf{BSC}}$ correspond to the `$Z$' random variables of the channel $W$, the channel $BEC$, and the channel $BSC$, respectively. Therefore, to show that there exists a BSC and a BEC satisfying \eqref{eq:same_E0}, it is sufficient to show that there exists $Z_{\mathsf{BEC}}$ and $z_{\mathsf{BSC}}$ random variables satisfying \eqref{eq:eq_Z}.
By the monotonicity results stated in Lemma 5, we know that 
\begin{align}
&g(\rho, z)\in [2^{-\rho}, 1], \quad \hbox{for } \rho\geq 0, \\
&g(\rho, z)\in [1, 2^{-\rho}], \quad \hbox{for }\rho\in(-1, 0],
\end{align}
for $z\in[0, 1]$. As a result, 
\begin{align}
&E[g(\rho, Z)]\in [2^{-\rho}, 1], \quad \hbox{for } \rho\geq 0, \\
&E[g(\rho, Z)]\in [1, 2^{-\rho}], \quad \hbox{for }\rho\in(-1, 0].
\end{align} 
Moreover, $g$ being continuous in $z$ for fixed values of $\rho$ implies that every intermediate value of the corresponding bounded interval 
will be taken by the function $g(\rho, z)$ for $z\in[0, 1]$, i.e. we can always find a $z^*\in[0, 1]$ such that 
\begin{equation}
E[g(\rho, Z)]=g(\rho, z^*).
\end{equation} 
Since, as indicated in Fact \ref{fact:Z_bec_bsc}, the random variable $Z_{BSC}$ of a
BSC is a constant $z_{BSC}$, the BSC defined in \eqref{eq:eq_Z} will be a BSC such that $z_{BSC} = z^*$. From this the crossover probability of the channel can be inferred.

To find a BEC which satisfies \eqref{eq:eq_Z}, we will use the BSC we have just defined with parameter $z^*$. 
Note that the extreme values of the bounded interval from which $g(\rho, z)$ takes values are given by $2^{-\rho} = g(\rho, 0)$ and $1 = g(\rho, 1)$. 
Moreover, the function $g$ being continuous in $z\in[0, 1]$ for fixed values of $\rho$, we can weight these two values with a probability distribution $p_0$ and $1-p_0$ such that 
\begin{equation}
g(\rho, z^*) = p_0 g(\rho, 0) + (1-p_0)g(\rho, 1).
\end{equation} 
Since, as indicated in Fact \ref{fact:Z_bec_bsc}, the random variable $Z_{BEC}$ of a BEC is $\{0, 1\}$ valued, the BEC defined in \eqref{eq:eq_Z} will be a BEC with erasure probability given by $P(Z_{BEC} = 0) = p_0$.    
\end{proof}

Upon this lemma, another property of BECs and BSCs is due:

\begin{fact}\label{fact:sweeping_bec_bsc}
The set of BECs and the set of BSCs both sweep all the possible values the $E_0$ parameters of B-DMCs can take at any $\rho >-1$. 
\end{fact}

Suppose now the $E_0$ curves of a BEC and a BSC intersect at a particular $\rho^*>-1$ other than $\rho^{*} = 0$. We would like to know if there are any other $\rho>-1$ values
apart from the trivial $\rho = 0$ such that the $E_0$ curves of these two channels intersect again? The next lemma answer this question. 

\begin{lemma}\label{lem::bsc_bec}
Suppose a BSC $BSC$, and a BEC $BEC$ satisfy
\begin{equation}\label{eq::Equal_E0_bsc_bec}
 E_0(\rho^{*}, BEC) = E_0(\rho^{*}, BSC),
\end{equation}
for some $\rho^{*} > -1$ such that $\rho^{*}\neq 0$. 
Then, if $\rho^*\leq 1$, there is only one other intersection point between the $E_0$ curves of the channels at $\rho=0$. If $\rho^{*} > 1$, the only intersection point in the interval $(-1,1]$ is once more at $\rho = 0$, and for the rest either the $E_0$ curves of the channels are tangent to each others at $\rho^*$, i.e.,
\begin{equation}\label{eq::R_tangent}
R(\rho^*, BEC) = R(\rho^*, BSC)
\end{equation}
is satisfied, or there exists a different $\rho'>1$ such that 
\begin{equation}\label{eq::Equal_again_bsc_bec}
 E_0(\rho', BEC) = E_0(\rho', BSC).
\end{equation} 

\end{lemma}
\begin{proof}
Let the erasure probability of the channel $BEC$ be $\epsilon$ and the channel $BSC$ be such that $z_{BSC} = z$. By \eqref{eq:E0_bec} and \eqref{eq:E0_bsc}, the condition for equality in \eqref{eq::Equal_E0_bsc_bec} translates into 
\begin{equation}\label{eq::cond_2}
g(\rho^{*}, z)  = 2^{-\rho^{*}}(1-\epsilon) + \epsilon. 
\end{equation}
Let the function $h(\rho, z)$ be defined as
\begin{equation}
 h(\rho, z) = \displaystyle\frac{g(\rho, z)-2^{-\rho}}{1-2^{-\rho}}.
\end{equation}
Observe that $h(\rho^*, z) = \epsilon$ and, in order for \eqref{eq::Equal_again_bsc_bec} to hold, we are looking for another $\rho'$ such that $h(\rho', z) = \epsilon$ holds. To find the answer, we need to study the monotonicity properties of the function $h(\rho, z)$ with respect to $\rho$. Indeed, one can show that the first derivative of $h(\rho, z)$ with respect to $\rho$ changes sign only once at $\rho_{\max}(z) \geq 3$ for every fixed value of $z$, such that $h(\rho, z)$ is increasing for $\rho\in(0, \rho_{\max}(z))$, 
and decreasing for $\rho>\rho_{\max}(z)$ with $\lim_{\rho\to\infty} h(\rho, z) = h(1, z)$. Consequently, if $\rho^*\in(-1,0)\cup(0, 1]$, no other $\rho'$ can satisfy \eqref{eq::Equal_again_bsc_bec}. On the other hand, if $\rho^*>1$, but $\rho^* \neq \rho_{\max}(z)$, then the two curves intersect twice. Finally, if $\rho^* =  \rho_{\max}(z)$, not only no other $\rho'$ can satisfy \eqref{eq::Equal_again_bsc_bec}, but also 
\begin{equation}
h(\rho^*, z) = h(\rho_{\max}(z), z) \geq h(\rho, z)
\end{equation}
holds for all $\rho >-1$. In this case, the $E_0$ curves of the channels will be tangent to each other, so \eqref{eq::R_tangent} holds as well.
As the analysis of the monotonicity property is tedious, we omit the proof.
\end{proof}

The previous lemma says that if the $E_0$ curves of a BEC and a BSC intersect somewhere between the interval $(-1, 0)\cup(0, 1]$, they cannot intersect a second time, except trivially at $0$, and if otherwise they intersect in the interval $(1, \infty)$, either the two curves are tangent to each other or they intersect twice in that interval, and the only intersection point in the interval $(-1, 1]$ is again at $0$. The significance of this lemma will become clear later when we interpret the extremality results. The lemma will help us to understand why some intervals of $\rho>-1$ are more interesting in the context of the extremality results presented in the main theorem.

\section{Extremality Results}\label{sec:ER}
In this section, we study the extremality of the BEC and the BSC with respect to the $E_0$ channel parameter. 
In particular, we show in Theorem \ref{thm:Extremality_Results} that a certain extremality property holds even when the quantities appearing in the parametric form of the random coding error exponent, i.e. $E_0$ and $E_0'$, are evaluated at different values of the parameter. The proof of the theorem is carried out in Section \ref{sub:Proof}.

\newtheorem{theorem}{Theorem}
\begin{theorem}\label{thm:Extremality_Results}
Given any fixed value of $\rho_1 > -1$, suppose a B-DMC $W$, a binary symmetric channel $BSC$, and a binary erasure channel $BEC$ satisfy 
\begin{equation}\label{eq::Equal_E0_thm1}
E_0(\rho_1, BSC) \overset{(a)}{\leq} E_0(\rho_1, W) \overset{(a')}{\leq} E_0(\rho_1, BEC), 
\end{equation}
for $\rho_1\neq 0$, or
\begin{equation}\label{eq::Equal_capacity_thm1}
\lim_{\rho\to0}
\displaystyle\frac{E_0(\rho, BSC)}{\rho}
 \overset{(a_0)}{\leq} \lim_{\rho\to 0}\displaystyle\frac{E_0(\rho, W)}{\rho}  \overset{({a_0}')}{\leq} \lim_{\rho\to 0}\displaystyle\frac{E_0(\rho, BEC)}{\rho},
\end{equation}
for $\rho_1 = 0$.

\begin{enumerate}
\item[(\textit{Part 1})]
If $\rho_1 \in [0, 3]$, then  
\begin{align} 
\label{eq:P1_R}&R(\rho_2, BSC) \overset{(b)}{\leq} R(\rho_2, W)\overset{(b')}{\leq} R(\rho_2, BEC), \\
\label{eq:P1_E}&E_0(\rho_2, BSC) \overset{(c)}{\leq} E_0(\rho_2, W)\overset{(c')}{\leq} E_0(\rho_2, BEC),
\end{align}
for any $\rho_2 \in [\rho_1, 3]$. 

\item[(\textit{Part 2})]
If $\rho_1 \in (-1, 0]$, then   
\begin{align}
\label{eq:P2_R}&R(\rho_2, BEC) \overset{(d)}{\leq} R(\rho_2, W)\overset{(d')}{\leq} R(\rho_2, BSC), \\
\label{eq:P2_E}&E_0(\rho_2, BSC) \overset{(e)}{\leq} E_0(\rho_2, W)\overset{(e')}{\leq} E_0(\rho_2, BEC),
\end{align}
for any $\rho_2\in(-1, \rho_1]$,

\item[(\textit{Part 3})]
If $\rho_1 \in (-1, 0]$, then  
\begin{equation} \label{eq:P3_E1}
E_0(\rho_2, BSC) \overset{(f)}{\leq} E_0(\rho_2, W)\overset{(f')}{\leq} E_0(\rho_2, BEC),
\end{equation}
for any $\rho_2 \geq 0$.

If $\rho_1 \in [0, 1]$, then  
\begin{equation} \label{eq:P3_E2}
E_0(\rho_2, BSC) \overset{(g)}{\leq} E_0(\rho_2, W)\overset{(g')}{\leq} E_0(\rho_2, BEC),
\end{equation}
for any $\rho_2 \geq \rho_1$.

If $\rho_1 > 1$, then  
\begin{equation} \label{eq:P3_E3}
E_0(\rho_2, BEC) \overset{(h)}{\leq} E_0(\rho_2, W)\overset{(h')}{\leq} E_0(\rho_2, BSC),
\end{equation}
for any $\rho_2 \in [0, 1]$.

If $\rho_1 > 1$, then  
\begin{equation} \label{eq:P3_E4}
E_0(\rho_2, BSC) \overset{(i)}{\leq} E_0(\rho_2, W)\overset{(i')}{\leq} E_0(\rho_2, BEC),
\end{equation}
for any $\rho_2 \in (-1, 0]$.
\end{enumerate}
Moreover, the extremalities hold with strict inequalities, except for $\rho_2 = 0$, whenever $(a)$ and $(a')$ in \eqref{eq::Equal_E0_thm1} are strict for $\rho_1\neq 0$, or $(a_0)$ and $(a_0')$ in \eqref{eq::Equal_capacity_thm1} are strict for $\rho_1 = 0$.
\end{theorem}

\begin{remark}\label{rem::r0}
In Theorem \ref{thm:Extremality_Results}, the inequalities $(a)$-$(a_0)$ imply the inequalities  
$(b)$, $(c)$, $(d)$, $(e)$, $(f)$, $(g)$, $(h)$, and $(i)$. Similarly, the inequalities $(a')$-$(a_0')$ imply the inequalities $(b')$ through $(i')$.
\end{remark}

\begin{remark}\label{rem:rho_3}
The value of ``$3$'' that appears in the interval in Part 1 of the theorem is a conservative estimate. The reader who follows the proof of Lemma \ref{lem:ftilde}, which is stated in Section \ref{sub:CL} and proved in Appendix II, will notice that this ``$3$'' may be replaced by a $\rho^*(W)$ that depends on the channel $W$. In the  
proof of Lemma \ref{lem:ftilde}, it is shown that $\rho^*(W)\geq 3$ for any $W$, but the lower bound is not necessarily tight. We chose the value 3 so as to not further complicate the statement of the theorem.
\end{remark}

For the special case where $\rho_1 = \rho_2 = \rho$, for $\rho\in[0,1]$, we recover in the next corollary, a result obtained in~\cite{notes1}.

\newtheorem{corollary}{Corollary}
\begin{corollary}[\cite{notes1}]\label{cor::Extremality_unpublished}
Given a symmetric B-DMC $W$, for any fixed value of $\rho \in [0, 1]$, find a binary symmetric channel $BSC$, and a binary erasure channel $BEC$ through
the equality 
\begin{equation}\label{eq::eq_R_cor1}
 R(\rho, W) = R(\rho, BEC) = R(\rho, BSC).  
\end{equation}
Then, 
\begin{align}
 \label{align::E0_order_cor1}&E_0(\rho, BEC) \leq E_0(\rho, W) \leq E_0(\rho, BSC), \\
 \label{align::Er_order_cor1}&E_r(\rho, BEC) \leq E_r(\rho, W) \leq E_r(\rho, BSC).  
\end{align}
\end{corollary}
\begin{proof}
Since $E_r(\rho,W)=E_0(\rho,W)-\rho R(\rho,W)$, it suffices to prove the
first set of inequalities in view of~\eqref{eq::eq_R_cor1}.
Taking $\rho_1=\rho_2=\rho$, \eqref{align::E0_order_cor1} holds by Theorem \ref{thm:Extremality_Results}. To see this, observe that had the channels on the contrary satisfied 
\begin{equation}
E_0(\rho, BSC) < E_0(\rho, W) < E_0(\rho, BEC),
\end{equation}
the results in Part 1 of the theorem would imply
\begin{equation}
R(\rho, BSC) < R(\rho, W) < R(\rho, BEC),
\end{equation}
contradicting the assumption \eqref{eq::eq_R_cor1} of the corollary. 
\end{proof}

Another particular case of Theorem \ref{thm:Extremality_Results} when $\rho_1=0$ recovers the result in~\cite{6034105}: amongst all symmetric B-DMCs of the same capacity, 
the BEC and the BSC are extremal with respect to the random coding exponent.

\begin{corollary}[Theorem 2.3~\cite{6034105}]\label{cor::Extremality_same_capacity}
Given a symmetric B-DMC $W$ of capacity $I(W)$, we define a binary symmetric channel $BSC$, and a binary erasure channel $BEC$ of the same capacity through 
the equality 
\begin{equation*}
 I(W) = I(BEC) = I(BSC).
\end{equation*}
Then, the random coding error exponent of the channels satisfy 
\begin{equation}\label{eq::Er_Extremality}
 E_r(R, BSC) \leq E_r(R, W) \leq E_r(R, BEC).  
\end{equation}
\end{corollary}
\begin{proof}
The equality of capacities is equivalent to 
\begin{equation*}
 \lim_{\rho\to 0}\displaystyle\frac{E_{0}(\rho, W)}{\rho} = \lim_{\rho\to 0}\displaystyle\frac{E_{0}(\rho, BEC)}{\rho} = \lim_{\rho\to 0}\displaystyle\frac{E_{0}(\rho, BSC)}{\rho}.
\end{equation*}
But in this case, we know by Part 1 in Theorem \ref{thm:Extremality_Results} that we have
\begin{equation}
E_0(\rho_2, BSC)\leq E_0(\rho_2, W)\leq E_0(\rho_2, BEC),
\end{equation}
for any $\rho_2\in[0,1]$. This, in turn, implies the inequality for the random coding exponent.
\end{proof}

Finally, note that in \cite{6377299} the above result of \cite{6034105} was extended to the region where $\rho>-1$. Namely, amongst all symmetric B-DMCs of the same capacity, the BEC and the BSC are extremal with
\begin{equation}
E_{0}(\rho, BSC) \leq E_{0}(\rho, W) \leq E_{0}(\rho, BEC),
\end{equation} 
for all $\rho>-1$. In particular, \cite[Theorem 1]{6377299} can also be recovered from Theorem \ref{thm:Extremality_Results}.

\subsection{Convexity Lemmas}\label{sub:CL}
The proof of Theorem \ref{thm:Extremality_Results} rests on the next two lemmas. The lemmas are proved in the Appendix.
\begin{lemma}\label{lem:ftilde}
For fixed values of $\rho_1, \rho_2\in\mathbf{R}\setminus\{-1\}$, we define the function $\tilde{f}_{\rho_1,\rho_2}(t)$ by
\begin{equation}
 \tilde{f}_{\rho_1,\rho_2}(t) = \displaystyle\frac{\partial}{\partial\rho_2} g(\rho_2, g^{-1}(\rho_1, t)),  
\end{equation}
for $t\in[2^{-\rho}, 1]\cup[1, 2^{-\rho}]$. Let $\tilde{f}_{\rho}(t)$ denotes the function when $\rho_1 = \rho_2 = \rho$. 
Then, $\tilde{f}_{\rho}(t)$ is a concave function in $t$ when $\rho\in(0, 3]$, convex when $\rho = (-1, 0]$ and $\rho\in(-\infty, -1)$.
Moreover, the function $\tilde{f}_{\rho_1,\rho_2}(t)$ is concave when $\rho_1, \rho_2\in[0, 1]$ such that $\rho_2 \geq \rho_1$.
\end{lemma}

\begin{lemma}\label{lem:f_concavity}
For fixed values of $\rho_1, \rho_2\in\mathbf{R}\setminus\{-1\}$, the function $f_{\rho_1,\rho_2}(t)$ defined as
\begin{equation}
 f_{\rho_1,\rho_2}(t) = g(\rho_2, g^{-1}(\rho_1, t)),  
\end{equation}
for $t\in[2^{-\rho}, 1]\cup[1, 2^{-\rho}]$, is concave in $t$ when $\rho_1\in(-1, 0]$ and $\rho_2 \geq 0$, when $\rho_1\in[0, 1]$ and $\rho_2 \geq \rho_1$, and when
$\rho_1 > 1$ and $\rho_2\in(-1, 0)$, and the function is convex when $\rho_1 > 1$ and $\rho_2\in(0, 1]$.
\end{lemma}

\subsection{Proof of Theorem \ref{thm:Extremality_Results}}\label{sub:Proof} 
Before we start proving the theorem's statement in its most general form, we will prove two particular cases of the theorem in the next two lemmas assuming $\rho_1=\rho_2=\rho$.

\begin{lemma}\label{lem::R_extr_equal_rho_1}
Given any fixed value of $\rho\in(0, 3)$, suppose a B-DMC $W$, a binary symmetric channel $BSC$, and a binary erasure channel $BEC$ satisfy
the equality 
\begin{equation}\label{eq::ineq_cond_E0}
E_0(\rho, BSC) \leq E_0(\rho, W) \leq E_0(\rho, BEC).
\end{equation}
Then, the following holds:
\begin{equation}\label{eq::R_extr_equal_rho_1}
 R(\rho, BSC) \leq R(\rho, W) \leq R(\rho, BEC),
\end{equation}
where the inequalities are strict if the inequalities in \eqref{eq::ineq_cond_E0} are strict.
\end{lemma}
\begin{proof}
Let us define another binary erasure channel $BEC^*$ and another binary symmetric channel $BSC^*$
through the following equality:
\begin{equation}\label{eq::eq_cond_E0}
E_0(\rho, BSC^*) = E_0(\rho, W) = E_0(\rho, BEC^*).
\end{equation}
Observe that by \eqref{eq:Eo}, the equality condition in Equation \eqref{eq::eq_cond_E0} is equivalent to the equality of
\begin{equation}\label{eq::eq_cond_g}
 \Expt\left[g(\rho, Z)\right] = \Expt\left[g(\rho, Z_{BEC^*})\right] = g(\rho, z_{BSC^*}). 
\end{equation}
Hence, the denominator in 
\begin{equation}\label{eq::R_denom}
 R(\rho, W) = \displaystyle\frac{\partial}{\partial\rho}E_0( \rho, W) = \displaystyle\frac{\Expt\left[-\partial g(\rho, Z)/\partial\rho\right]}{\Expt\left[g(\rho, Z)\right]}
\end{equation}
is the same for the three channels. Then, the proof can be completed using the concavity of the function $\tilde{f}_{\rho}(t)$ in $t$ for $\rho\in(0, 3]$, which was shown in Lemma \ref{lem:ftilde}, and the special structure of the $Z$ random variable of a BEC and a BSC. To see this, let us define the random variable $T = g(\rho, Z)\in[2^{-\rho}, 1]$.
Then, we note that $\tilde{f}_{\rho}(T) = \partial g(\rho, Z)/\partial\rho$, and $E[T]$ gives \eqref{eq::eq_cond_g}. So,
\begin{equation}
R(\rho, W) = \displaystyle\frac{\Expt\left[\tilde{f}_{\rho}(T)\right]}{E[T]}, \quad R(\rho, BSC^*) = \displaystyle\frac{\tilde{f}_{\rho}(E[T])}{E[T]}. 
\end{equation}
To derive the expression for $R(\rho, BEC^*)$, recall by Fact \ref{fact:Z_bec_bsc} that $Z_{BEC^*} = \{0, 1\}$. Using $E[T]=\Expt\left[g(\rho, Z_{BEC^*})\right]$, we get
\begin{equation}
P(Z_{BEC^*} = 0) = \displaystyle\frac{E[T]-1}{2^{-\rho}-1}.
\end{equation}
Hence,
\begin{equation}
R(\rho, BEC^*) = \displaystyle\frac{\tilde{f}_{\rho}(2^{-\rho}) P(Z_{BEC^*} = 0) + \tilde{f}_{\rho}(1)P(Z_{BEC^*} = 1)}{E[T]}.
\end{equation}
Now, by the two sides of the Jensen's inequality for concave functions we have
\begin{equation}\label{eq::jensen}
\tilde{f}_{\rho}(1) + \displaystyle\frac{\tilde{f}_{\rho}(1) - \tilde{f}_{\rho}(2^{-\rho})}{1-2^{-\rho}} \left( \Expt \left[T \right] - 1 \right) \leq \Expt\left[\tilde{f}_{\rho}(T)\right] 
\leq \tilde{f}_{\rho}(\Expt\left[T \right]).
\end{equation}
Dividing all sides by $\Expt\left[T\right]>0$ and negating the expressions in \eqref{eq::jensen}, we get
\begin{equation}\label{eq:final_step}
 R(\rho, BSC^*) \leq R(\rho, W) \leq R(\rho, BEC^*). 
\end{equation}
The final step of the proof is to show \eqref{eq:final_step} implies \eqref{eq::R_extr_equal_rho_1}. For that purpose, recall that by Fact 2 that the set of BSCs and the set of BECs are strictly ordered in their $E_0$ and $R$ parameters for $\rho\in(0, 3]$. As we have
\begin{align}
\label{align:ineq_str1}&E_0(\rho, BSC) \leq E_0(\rho, BSC^*), \\
\label{align:ineq_str2}&E_0(\rho, BEC^*) \leq E_0(\rho, BEC),
\end{align}
we conclude by Lemma \ref{lem:ordering_bec_bsc} that
\begin{align}
\label{align:ineq_str3}&R(\rho, BSC) \leq R(\rho, BSC^*), \\
\label{align:ineq_str4}&R(\rho, BEC^*) \leq R(\rho, BEC)
\end{align}
holds for $\rho \geq 0$. From this \eqref{eq::R_extr_equal_rho_1} follows. Moreover, if the inequalities in \eqref{eq::ineq_cond_E0} are strict than the ones in \eqref{align:ineq_str1} and \eqref{align:ineq_str2}, and thus, \eqref{align:ineq_str3} and \eqref{align:ineq_str4} are strict as well. Consequently, the inequalities in \eqref{eq::R_extr_equal_rho_1} hold strictly as claimed.
\end{proof}

\begin{remark}\label{rem:r4}
Note that Lemma \ref{lem::R_extr_equal_rho_1} and Corollary \ref{cor::Extremality_unpublished} are of the same flavor.
Indeed, one can easily derive one from the other using the degradation argument discussed in Fact 2. So, the result of~\cite{notes1} could also have been used to characterize the behavior of the $E_0$ curves for the $\rho\in(0, 1]$ interval.
However, the proofs of the lemma and the corollary are different as they involve different convexity analysis.
\end{remark}

\begin{lemma}\label{lem::R_extr_equal_rho_2}
Given any fixed value of $\rho\in(-1, 0)$, suppose a B-DMC $W$, a binary symmetric channel $BSC$, and a binary erasure channel $BEC$ satisfy the condition  \eqref{eq::ineq_cond_E0} of Lemma \ref{lem::R_extr_equal_rho_1}. Then, the following holds:
\begin{equation}\label{eq::R_extr_equal_rho_2}
 R(\rho, BEC) \leq R(\rho, W) \leq R(\rho, BSC),
\end{equation}
where the inequalities are strict if the inequalities in \eqref{eq::ineq_cond_E0} are strict.
\end{lemma}
\begin{proof}
Let $BEC^*$ and $BSC^*$ be as defined in the proof of Lemma \ref{lem::R_extr_equal_rho_1}. Once again, the equality condition in \eqref{eq::eq_cond_E0} implies  the denominator in \eqref{eq::R_denom} is the same for the three channels. Then, the inequalities
\begin{equation}
 R(\rho, BEC^*) \leq R(\rho, W) \leq R(\rho, BSC^*)
\end{equation}
follow using the convexity of the function $\tilde{f}_{\rho}(t)$ in $t$ when $\rho\in(-1, 0]$, which was shown in Lemma \ref{lem:ftilde}, and applying Jensen's inequalities. 
Finally, since $E_0(\rho, BSC) \leq E_0(\rho, BSC^*)$ and $E_0(\rho, BEC^*) \leq E_0(\rho, BEC)$, we know by Fact 2 that these BSCs and BECs are ordered by degradation, and we conclude by Lemma \ref{lem:ordering_bec_bsc} that 
we have $R(\rho, BSC^*) \leq R(\rho, BSC)$ and $R(\rho, BEC) \leq R(\rho, BEC^*)$, for $\rho \in (-1, 0]$. From this \eqref{eq::R_extr_equal_rho_2} follows. The claim about the strictness of the inequalities can be proved similarly as in the proof of Lemma \ref{lem::R_extr_equal_rho_1}.
\end{proof}

Now, we are ready to prove the theorem.

\begin{proof}[Proof of Theorem \ref{thm:Extremality_Results}]

We will first prove the claims for $\rho_1\in(-1, 0)\cup(0, \infty)$, leaving the case $\rho_1 = 0$ to the last. In fact, we will show that the results proved for $\rho_1\in(-1, 0)\cup(0, \infty)$ will immediately extend to $\rho_1 = 0$ by the continuity of $E_0$ in its arguments. 

We start by proving the inequalities \eqref{eq:P1_R} and \eqref{eq:P1_E} in Part 1 for the case $\rho_1\in(0, 3]$. By Lemma \ref{lem::R_extr_equal_rho_1}, we know that \eqref{eq:P1_R} holds for $\rho_2 = \rho_1$. So, we only need to prove the theorem for $\rho_2\in(0, 3]$ such that
$\rho_2>\rho_1$. 
By the continuity of $E_0(\rho, BEC)$ and $E_0(\rho, BSC)$ in the channels' erasure and crossover probabilities, respectively, it suffices to show that 
\begin{equation}
E_0(\rho_1, BSC) < E_0(\rho_1, W) < E_0(\rho_1, BEC) 
\end{equation}
implies 
\begin{equation}\label{eq:ineq1}
 E_0(\rho_2, BSC) < E_0(\rho_2, W) < E_0(\rho_2, BEC).
\end{equation}
Then, Lemma \ref{lem::R_extr_equal_rho_1} will imply
\begin{equation}
 R(\rho_2, BSC) < R(\rho_2, W) < R(\rho_2, BEC).
\end{equation}

We define $D(\rho) = E_0(\rho, W) - E_0(\rho, BEC)$. Let $D'(\rho)$ denotes the first derivative of $D(\rho)$ with respect to $\rho$.
Noting that $R(\rho, W) = \displaystyle\frac{\partial}{\partial\rho}E_0(\rho, W)$, the inequality in \eqref{eq:ineq1} is implied by the following statement: 
\begin{equation}
 D(\rho_1) < 0 \quad \hbox{and by Lemma \ref{lem::R_extr_equal_rho_1}} \quad (D(\rho) < 0 \Rightarrow D'(\rho) < 0) \quad \Rightarrow \quad D(\rho_2) < 0.
\end{equation}
But this is true by elementary considerations on differential equations. Indeed, suppose to the contrary that
\begin{equation}
 D(\rho_1) < 0, \quad \hbox{and} \quad (D(\rho) < 0 \Rightarrow D'(\rho) < 0), \quad \hbox{but} \quad D(\rho_2) \geq 0.
\end{equation}
Then, there exists $\rho_1 < \rho_3 \leq \rho_2$ such that $D(\rho) < 0$, for $\forall\rho\in[\rho_1, \rho_3)$, and $D(\rho_3) = 0$.
But then there exists $\rho_1 < \rho_4 < \rho_3$ such that
\begin{equation}
D'(\rho_4) = \displaystyle\frac{D(\rho_3) - D(\rho_1)}{\rho_3 - \rho_1} > 0, 
\end{equation}
and $D(\rho_4) < 0$, contradicting the assumption.   

The inequality for the BSC can be obtained similarly by letting $D(\rho) = E_0(\rho, BSC) - E_0(\rho, W)$ and applying the above argument once more.  

We continue with the proof of the inequalities in \eqref{eq:P2_R} and \eqref{eq:P2_E} in Part 2 for the case $\rho_1\in(-1, 0)$. The proof follows along the same lines of the previous part. By Lemma \ref{lem::R_extr_equal_rho_2}, we know that the inequalities in \eqref{eq:P2_R} hold for $\rho_2 = \rho_1$. So, we only need to prove the theorem for $\rho_2 < \rho_1$. By the continuity of $E_0(\rho, BEC)$ and $E_0(\rho, BSC)$ in the channels' erasure and crossover probabilities, respectively, , it suffices to show that 
\begin{equation*}
E_0(\rho_1, BSC) < E_0(\rho_1, W) < E_0(\rho_1, BEC) 
\end{equation*}
implies 
\begin{equation*}
 E_0(\rho_2, BSC) < E_0(\rho_2, W) < E_0(\rho_2, BEC).
\end{equation*}
Then, Lemma \ref{lem::R_extr_equal_rho_2} will imply 
\begin{equation*}
 R(\rho_2, BEC) < R(\rho_2, W) < R(\rho_2, BSC).
\end{equation*}
We define $D(\rho) = E_0(\rho, W) - E_0(\rho, BEC)$.
Noting that $R(\rho) = \displaystyle\frac{\partial}{\partial\rho}E_0(\rho)$, the corollary is implied by the following statement: 
\begin{equation*}
 D(\rho_1) < 0 \quad \hbox{and by Lemma \ref{lem::R_extr_equal_rho_2}} \quad (D(\rho) < 0 \Rightarrow D'(\rho) > 0) \quad \Rightarrow \quad D(\rho_2) < 0.
\end{equation*}
But this is true by an analogous reasoning as before.

The inequality for the BSC can be obtained similarly by letting $D(\rho) = E_0(\rho, BSC) - E_0(\rho, W)$ and applying the above argument once more.
This concludes the proof of Part 2.

For Part 3, we will only do the proof of \eqref{eq:P3_E1} for the case $\rho_1\in(-1, 0)$ and $\rho_2 \geq 0$ as all the other claims can be proved in the same way using the convexity properties of the function $f_{\rho_1,\rho_2}(t)$ discussed in Lemma \ref{lem:f_concavity}. 

Let $T = g(\rho_1, Z)$. We know that the condition in \eqref{eq::Equal_E0_thm1} is equivalent to 
\begin{equation}
\Expt\left[g(\rho_1, Z_{BEC})\right] \leq \Expt\left[g(\rho_1, Z)\right] \leq g(\rho_1, z_{BSC}).
\end{equation}
Define the BEC $BEC^*$ and the BSC $BSC^*$ through the equality
\begin{equation}
 \Expt\left[g(\rho_1, Z)\right] = \Expt\left[g(\rho_1, Z_{BEC^*})\right] = g(\rho_1, z_{BSC^*}). 
\end{equation}
As by Lemma \ref{lem:f_concavity} we know the function $f_{\rho_1,\rho_2}(t)$ is concave in $t$ when $\rho_1\in(-1, 0]$ and $\rho_2\geq 0$,
we can apply the two sides of Jensen's inequality to obtain
\begin{equation}
\Expt\left[f_{\rho_1,\rho_2}(g(\rho_1, Z_{BEC^*}))\right] \leq \Expt\left[f_{\rho_1,\rho_2}(T)\right] \leq f_{\rho_1,\rho_2}(g(\rho_1, Z_{BSC^*})),
\end{equation}
which is equivalent to 
\begin{equation}
\Expt\left[g(\rho_2, Z_{BEC^*})\right] \leq \Expt\left[g(\rho_2, Z)\right] \leq g(\rho_2, z_{BSC^*}).
\end{equation}

To get the claimed inequalities in \eqref{eq:P3_E1}, we simply need to use the ordering argument based on Fact 2 for the two BECs and the two BSCs. As we have illustrated this argument before in the proof of Lemma \ref{lem::R_extr_equal_rho_1}, we do not repeat it here.

The last step is to prove the theorem for the case $\rho_1=0$. We will only present the proof extension for the inequalities $(b')$ and $(c')$ in Part 1 as the same argument can be used to extend all the remaining results. Moreover, once again by the continuity of $E_0(\rho, BEC)$ in the channels' erasure probability, it suffices to show the results assuming $(a_0')$ in \eqref{eq::Equal_capacity_thm1} holds with strict inequality. 

So, we assume the given channels $W$ and $BEC$ satisfy $I(W) < I(BEC)$. Then,
\begin{equation}
\displaystyle\lim_{\rho\to 0^+} \displaystyle\frac{E_0(\rho, W) - E_0(\rho, BEC)}{\rho}  = I(W)-I(BEC) < 0.
\end{equation}
(We assumed $\rho\to0^+$ for simplicity as the above limit for $\rho\to 0$ is well defined). Hence, for any sufficiently small $\rho>0$, we have
\begin{equation}
E_0(\rho, W) < E_0(\rho, BEC).
\end{equation}
Moreover, we already proved that this implies
\begin{equation}
E_0(\rho_2, W) \leq E_0(\rho_2, BEC).
\end{equation}
for all $\rho_2\in[\rho, 3]$. As $\rho>0$ is arbitrary, we conclude the result should hold for all $\rho_2\in[0, 3]$.

Now, we can carry the proof as follows. First, we let $\epsilon\in[0, 1]$ be the erasure probability of the BEC $BEC_\epsilon$ which satisfies $I(W_{\epsilon}) = I(W)$. Then, we take a sequence of BECs $BEC_{\epsilon_n}$ of erasure probabilities $\epsilon_n\in[0, 1]$ such that the sequence $\epsilon_n$ is increasing to $\epsilon$. In this case, we know that
\begin{equation}
I(W) < I(BEC_{\epsilon_n}).
\end{equation}
By the previous argument, we conclude that for all the channels $BEC_{\epsilon_n}$,
\begin{equation}
E_0(\rho_2, W) \leq E_0(\rho_2, BEC_{\epsilon_n})
\end{equation}
holds for all  $\rho_2\in[0, 3]$. Taking the limit for the sequence $\epsilon_n$, we conclude by continuity that the result also holds for the channel $BEC_\epsilon$, i.e.,
\begin{equation}
E_0(\rho_2, W) \leq E_0(\rho_2, BEC_{\epsilon})
\end{equation}
holds for $\rho_2\in[0, 3]$. As the ordering $E_0(\rho_2, BEC_{\epsilon}) \leq E_0(\rho_2, BEC)$ holds, the inequality $(c')$ in \eqref{eq:P1_E} is proved.  By Lemma \ref{lem::R_extr_equal_rho_1}, the inequality $(b')$ follows. 
\end{proof}

\subsection{Extremality of R\'{e}nyi Entropies}
In this section, we show how the results of Theorem \ref{thm:Extremality_Results} can be translated into extremalities for R\'{e}nyi entropies using the definition given in \eqref{eq:Eo_over_ro}.

Observe that the assumption in \eqref{eq::Equal_E0_thm1} of Theorem \ref{thm:Extremality_Results} can be equivalently stated as
\begin{equation*}
\displaystyle\frac{E_0(\rho_1, BSC)}{\rho_1} \leq  \displaystyle\frac{E_0(\rho_1, W)}{\rho_1}  \leq \displaystyle\frac{E_0(\rho_1, BEC)}{\rho_1},
\end{equation*} 
for $\rho_1 > 0$, and
\begin{equation*}
\displaystyle\frac{E_0(\rho_1, BEC)}{\rho_1} \leq  \displaystyle\frac{E_0(\rho_1, W)}{\rho_1}  \leq \displaystyle\frac{E_0(\rho_1, BSC)}{\rho_1},
\end{equation*} 
for $\rho_1\in(-1, 0)$. Note that by Lemma \ref{lem:ordering_bec_bsc}, while for $\rho_1 > 0$ a worst BEC and a worst BSC has a smaller $E_0$ parameter, for $\rho\in(-1, 0)$ the opposite is true.
Consequently, all the results obtained for the parameter $E_0(\rho, W)$ can be restated in terms of R\'{e}nyi entropies via \eqref{eq:Eo_over_ro}. For the sake of brevity, we will only restate in the next corollary the result given in \eqref{eq:P3_E2} in Part 3 of the theorem in terms of R\'{e}nyi entropies.

\begin{corollary}
Given a binary uniform random variable $X$, among all jointly distibuted random variables $(X, Y)$ of equal R\'{e}nyi equivocation $H_{\alpha}(X \mid Y)$ of order $\alpha\in(0, 1/2]$, the 
R\'{e}nyi equivocation of order $\beta \geq 0$ such that $\beta \geq \alpha$ is maximized when $X$ and $Y$ are coupled by a BEC, and minimized when coupled by a BSC. For $\beta \leq \alpha$ values,
the maximizing and minimizing distributions are reversed.
\end{corollary}
\begin{proof}
Recall that $\alpha = 1/(1+\rho)$. So for $\alpha\in(0, 1/2]$, we have $\rho\in[0, 1]$. Moreover, $\alpha$ is decreasing with $\rho$. Hence, the inequalities for $\beta \leq \alpha$ and for $\beta \geq \alpha$ follow directly from \eqref{eq:P3_E2} in Part 3 of Theorem \ref{thm:Extremality_Results} using the definition given in equation \eqref{eq:Eo_over_ro} together with the fact that $H_{\alpha}(X) = 1$ under the uniform distribution.  
\end{proof}

\subsection{Graphical Interpretation of the Extremality Results}
\begin{figure}
\centering
\input{fig1}
\caption{Extremality of $E_0(\rho)$ when the channels intersect at $\rho_0\in(-1, 0)$. Dashed line: BEC(0.3) \& Solid line: BSC(0.1102).}

\input{fig2}
\centering
\caption{Extremality of $E_0(\rho)$ when the channels have equal capacity 0.5. Dashed line: BEC(0.5) \& Solid line: BSC(0.1102).}

\centering
\input{fig3}
\caption{Extremality of $E_0(\rho)$ when the channels have equal cut-off rate. Dashed line: BEC(0.626278) \& Solid line: BSC(0.1102).}
\end{figure}

\begin{figure}
\input{fig4}
\centering
\caption{Extremality of $E_0(\rho)$ when the channels have equal $E_0(\rho^{*})$ and equal rate at $\rho^{*}> 1$. Dashed line: BEC(0.6777) \& Solid line: BSC(0.1102).}

\centering
\input{fig5}
\caption{Extremality of $E_0(\rho)$ when the channels intersect at $\rho_0 > 1$. Dashed line: BEC(0.67) \& Solid line: BSC(0.1102).}

\end{figure}
In this section, we provide a graphical interpretation of the theorem and the corollaries through Figures 1 to 5. Suppose that the $E_0$ curves of a given B-DMC, a BEC, and a BSC pass through a given point $(\rho_0, e_0)$, for some $\rho_0 > -1$.

By the results stated in \eqref{eq:P2_E} and \eqref{eq:P3_E1} of Theorem \ref{thm:Extremality_Results}, we know that when $\rho_0\in(-1, 0)$, then these curves do not intersect again except at $\rho = 0$, and the BEC and BSC always remain extremal even though their extremal behaviour get reversed after the intersection points. Figure 1 illustrates this relation. 

A special case where the $E_0$ curves of the BEC and the BSC remain extremal for the entire $\rho > -1$ region, and with no reversal, corresponds to channels of the same capacity; 
as discussed after Corollary \ref{cor::Extremality_same_capacity}, Theorem \ref{thm:Extremality_Results} shows that the $E_0$ curves of these channels are upper bounded by the BEC's curve and lower bounded by the BSC's one. Figure 2 illustrates this relation. 

Another situation where the $E_0$ curves of the BEC and the BSC exhibit extremality for the entire region $\rho > -1$ occurs when $\rho_0\in(0, 1]$; \eqref{eq:P1_E} and \eqref{eq:P3_E2} of Theorem \ref{thm:Extremality_Results} imply the BEC and the BSC will be $E_0$ extremal, one again with the extremalities reversed after the intersections.
Figure 3 illustrates this relation. 

Now, we consider the case when $\rho_0 > 1$. By Theorem \ref{thm:Extremality_Results}, we know the curves only intersect at $\rho = 0$ in the interval $\rho\in(-1, 1]$, and the BEC and the BSC are extremal in $(-1, 0)$ and $(0, 1)$ with reversed extremalities. 
Although the thoerem provides a partial result, it is not clear what happens in the interval $\rho > 1$. It turns out that the BEC and the BSC are no longer extremal for $\rho > 1$ in general. We will show this result by studying the intersection points of the $E_0$ curves of a given BSC with different BECs using Lemma \ref{lem::bsc_bec}.

Suppose a BEC $BEC$ and a BSC $BSC$ satisfy 
\begin{align}
&E_0(\rho^*, BEC) = E_0(\rho^*, BSC), \\
&R(\rho^*, BEC) = R(\rho^*, BSC), 
\end{align}
for a particular $\rho^* > 1$. We know by Lemma \ref{lem::bsc_bec} that this corresponds to the case the $E_0$ curves of these two channels are tangent at $\rho^{*} > 1$ and do not intersect at any other point except $\rho = 0$. Moreover, by Theorem \ref{thm:Extremality_Results}, we know the capacities of the channels are such that $I(BEC) \leq I(BSC)$. Figure 4 illustrates this relation.

Suppose the erasure probability of the BEC channel is increased. By the ordering we discussed in Fact 2, it is not difficult to see that the $E_0$ curves of the BSC and
that BEC will not intersect at any point other than $\rho = 0$. On the other hand, assume instead the erasure probability of the channel is decreased such that the capacity of the new BEC is still smaller than the capacity of the BSC. In this case, as long as the cut-off rate of the BSC is larger than the cut-off rate of the BEC,
the BSC and the new BECs will intersect twice after $\rho = 0$, first in the interval $(1, \rho^{*})$, then after $\rho^{*}$. Figure 5 illustrates this relation.
Once the cut-off rate of the BEC becomes larger than that of the BSC, we are back at the situation where the intersection point falls in the interval $[0, 1]$, and we recover the general extremality result we have already discussed. Then, we can keep decreasing the erasure probability until the BEC and the BSC have the same capacity to recover another special case. Finally, decreasing more the erasure probability, until there is no other intersection anywhere except at $\rho = 0$, will cause the $E_0$ curves of the BSC and the new BECs to intersect in the interval $(-1, 0)$, in which case once more the BSC and the BECs will be $E_0$ extremal for the entire $\rho > -1$ region.

The analysis above shows us that most of the BECs and the BSCs whose $E_0$ curves intersect in the interval $\rho>1$ have two intersection points in that interval. In such a case, the BEC and the BSC are no longer extremal as we do not expect a class of B-DMCs $\mathcal{W}$ which satisfy for all $W\in\mathcal{W}$ the equality
\begin{equation}
E_0(\rho_0, W) = E_0(\rho_0, BEC) = E_0(\rho_0, BSC),
\end{equation} 
for any fixed $\rho_0 > 1$, to intersect a second time at the same point where the BEC and the BSC intersect the second time in the interval $(1, \infty)$.

\section{Conclusions} 
We have described certain extremalities for B-DMCs when the information measure is Gallager's $E_0$ evaluated under the uniform input distribution. These properties yield in straightforward fashion recent results by Fabregas et al.\cite{6034105}, \cite{6377299}, and also extremal properties for the R\'{e}nyi entropies.

Finally, it is worth emphasizing that all the conclusions of the paper are valid for arbitrary binary input channels as long as one evaluates all the quantities under the uniform input distribution. 

\section*{Acknowledgment}
The author would like to thank Emre Telatar for helpful discussions.
This work was supported by Swiss National Science Foundation under grant number 200021-125347/1.

\section*{Appendices}
The Appendices contain four parts. In the first three of them, we prove 
Lemma \ref{lem:g_concavity}, Lemma \ref{lem:ftilde}, and Lemma \ref{lem:f_concavity}, respectively. The final part proves two other lemmas needed in these proofs.

\subsection*{Appendix I}
\begin{proof}[Proof of Lemma \ref{lem:g_concavity}]
Taking the first derivative of \eqref{eq:g} with respect to $z$, we get
\begin{align}
\frac{\partial g(\rho, z)}{\partial z} 
&=  \left( \frac{1}{2}(1 + z)^{\frac{1}{1+\rho}} + \frac{1}{2}(1 - z)^{\frac{1}{1+\rho}} \right)^{\rho}  \left( \frac{1}{2}(1 + z)^{\frac{-\rho}{1+\rho}} - \frac{1}{2}(1 - z)^{\frac{-\rho}{1+\rho}} \right) \nonumber \\
&= \underbrace{\left(\frac{1}{2}\right)^{1+\rho} \left(1 + \left( \frac{1 - z}{1 + z}\right) ^{\frac{1}{1+\rho}}\right)^{\rho}}_{\geq 0} \left( 1 - \left( \frac{1 - z}{1 + z}\right) ^{\frac{-\rho}{1+\rho}}\right) \label{eq::dgu}.
\end{align}
As we have 
\begin{equation*}
\displaystyle\frac{1 - z}{1 + z} \leq 1,  
\end{equation*}
for $\forall z\in[0, 1]$, the monotonicity claims follow by noting that when $\rho\in(-\infty, -1)\cup[0, \infty)$:
\begin{equation*}
\frac{\rho}{1+\rho} \geq 0 \quad \Rightarrow \quad \left( 1 - \left( \frac{1 - z}{1 + z}\right) ^{\frac{-\rho}{1+\rho}}\right) \leq 0 \quad \Rightarrow \quad \frac{\partial g(\rho, z)}{\partial z} \leq 0,
\end{equation*}
and when $\rho\in(-1, 0]$:
\begin{equation*}
 \frac{\rho}{1+\rho} \leq 0 \quad \Rightarrow \quad \left( 1 - \left( \frac{1 - z}{1 + z}\right) ^{\frac{-\rho}{1+\rho}}\right) \geq 0 \quad \Rightarrow \quad \frac{\partial g(\rho, z)}{\partial z} \geq 0. 
\end{equation*}

Taking the second derivative with respect to $z$, we get
\begin{equation}
 \frac{\partial^{2} g(\rho, z)}{\partial z^{2}} = - \frac{\rho}{1+\rho} \underbrace{\left( 1 - z^{2}\right)^{\frac{1}{1+\rho}-2} \left(  \frac{1}{2}(1 + z)^{\frac{1}{1+\rho}} + \frac{1}{2}(1 - z)^{\frac{1}{1+\rho}}  \right) ^{-1+\rho}}_{\geq 0}. \nonumber
\end{equation}
The convexity claims follow once again by inspecting the sign of $\displaystyle\frac{\rho}{1+\rho}$ in different intervals, i.e. when $\rho\in(-\infty, -1)\cup[0, \infty)$:
\begin{equation*}
\frac{\rho}{1+\rho} \geq 0 \quad \Rightarrow \quad \frac{\partial^{2} g(\rho, z)}{\partial z^{2}} \leq 0, 
\end{equation*}
and when $\rho\in(-1, 0]$:
\begin{equation*}
 \frac{\rho}{1+\rho} \leq 0 \quad \Rightarrow \quad \frac{\partial^{2} g(\rho, z)}{\partial z^{2}} \geq 0. 
\end{equation*}
\end{proof}

\subsection*{Appendix II} 
\begin{proof}[Proof of Lemma \ref{lem:ftilde}]
We begin by introducing some definitions to simplify notations. Let 
\begin{equation}
g'(\rho, z) = \displaystyle\frac{\partial g(\rho, z)}{\partial z}.
\end{equation}
We define
\begin{align}
 \label{eq:h} h(z) &= \frac{1 - z}{1 + z},  \\
 \label{eq:alpha}\alpha(\rho, z) &= (1 +  h(z)^{\frac{1}{1+\rho}})^{\rho}, \\
 \label{eq:beta}\beta(\rho, z) &= (1 - h(z)^{\frac{-\rho}{1+\rho}}),
\end{align}
for $z\in[0, 1]$, $\rho\in\mathbf{R}\setminus\{-1\}$. By equation \eqref{eq::dgu} in Lemma \ref{lem:g_concavity}, we have
\begin{equation}
 g'(\rho, z)=  \left(\frac{1}{2}\right)^{1+\rho} \alpha(\rho, z) \beta(\rho, z). 
\end{equation}

Taking the first derivative of $\tilde{f}_{\rho_1,\rho_2}(t)$ with respect to $t$, we obtain 
\begin{align}
\frac{\partial \tilde{f}_{\rho_1,\rho_2}(t)}{\partial t} &= \frac{\partial}{\partial t} \displaystyle\frac{\partial}{\partial\rho_2} g(\rho_2, g^{-1}(\rho_1, t)) \\
&= \frac{\partial}{\partial \rho_2} \displaystyle\frac{\partial}{\partial t} g(\rho_2, g^{-1}(\rho_1, t)) \\
&= \frac{\partial}{\partial \rho_2} \frac{g'(\rho_2, g^{-1}(\rho_1, t))}{g'(\rho_1, g^{-1}(\rho_1, t))}. 
\end{align}
Let $z = g^{-1}(\rho_1, t)$. As $g(\rho, z)$ is a monotone function in $z$ by Lemma \ref{lem:g_concavity} in Appendix I, so is $z = g^{-1}(\rho, t)$ in $t$. 
Hence, we can check the convexity of $\tilde{f}_{\rho_1,\rho_2}(t)$ with respect to $t$ from the monotonicity with respect to $z$ of the following expression:
\begin{align}
\frac{\partial}{\partial \rho_2} \frac{g'(\rho_2, z)}{g'(\rho_1, z)} &= \frac{\partial}{\partial \rho_2}  2^{\rho_1-\rho_2}  \frac{\alpha(\rho_2, z) \beta(\rho_2, z)}{\alpha(\rho_1, z) \beta(\rho_1, z)} \nonumber \\
\label{eq::mono}&= \frac{ 2^{-\rho_2} \alpha(\rho_2, z) \beta(\rho_2, z)}{ 2^{-\rho_1} \alpha(\rho_1, z) \beta(\rho_1, z)} \left( \frac{ \partial 2^{-\rho_2} \alpha(\rho_2, z)/\partial \rho_2}{ 2^{-\rho_2} \alpha(\rho_2, z)}
 + \frac{ \partial \beta(\rho_2, z )/\partial \rho_2}{ \beta(\rho_2, z)}\right)  
\end{align}
where
\begin{align}
 \displaystyle\frac{\partial 2^{-\rho_2} \alpha(\rho_2, z)}{\partial \rho_2} &=  \frac{\partial}{\partial \rho_2}  \left( \frac{1}{2} + \frac{1}{2}h(z)^{\frac{1}{1+\rho_2}}\right) ^{\rho_2} \\
&=  \left( \frac{1}{2} + \frac{1}{2}h(z)^{\frac{1}{1+\rho_2}}\right) ^{\rho_2} \left( \log{\left( \frac{1}{2} + \frac{1}{2}h(z)^{\frac{1}{1+\rho_2}}\right) } + \rho_2 \frac{\frac{1}{2} h(z)^{\frac{1}{1+\rho_2}} \frac{-1}{(1+\rho_2)^2} \log{h(z)}}{\frac{1}{2} + \frac{1}{2}h(z)^{\frac{1}{1+\rho_2}}}\right) \\ 
& = 2^{-\rho_2} \alpha(\rho_2, z) \left( \log{\left( \frac{1}{2} + \frac{1}{2}h(z)^{\frac{1}{1+\rho_2}}\right) } - \frac{\rho_2 h(z)^{\frac{1}{1+\rho_2}}\log{h(z)}}{\left(1+\rho_2\right)^2 \left(1+h(z)^{\frac{1}{1+\rho_2}}\right)}\right),  
\end{align}
and
\begin{align}
 \displaystyle\frac{\partial \beta(\rho_2, z)}{\partial\rho_2}  &=  \frac{\partial}{\partial\rho_2} \left( 1 - h(z)^{\frac{-\rho_2}{1+\rho_2}}\right)  \\
&= \frac{1}{(1+\rho_2)^2} h(z)^{\frac{-\rho_2}{1+\rho_2}} \log{h(z)}.
\end{align}

Hence, the expression inside the parenthesis in \eqref{eq::mono} equals
\begin{align}
& \log{\left( \frac{1}{2} + \frac{1}{2}h(z)^{\frac{1}{1+\rho_2}}\right) } - \frac{\rho_2 h(z)^{\frac{1}{1+\rho_2}}\log{h(z)} }{\left(1+\rho_2\right)^2 \left(1+h(z)^{\frac{1}{1+\rho_2}}\right)} 
+ \frac{ h(z)^{\frac{-\rho_2}{1+\rho_2}} \log{h(z)}}{ (1+\rho_2)^2  \left( 1 - h(z)^{\frac{-\rho_2}{1+\rho_2}}\right)}  \\
=& \log{\left( \frac{1}{2} + \frac{1}{2}h(z)^{\frac{1}{1+\rho_2}}\right) } - \frac{\rho_2 h(z)^{\frac{1}{1+\rho_2}}\log{h(z)} }{\left(1+\rho_2\right)^2 \left(1+h(z)^{\frac{1}{1+\rho_2}}\right)} 
+ \frac{ \log{h(z)}}{ (1+\rho_2)^2  \left(h(z)^{\frac{\rho_2}{1+\rho_2}} - 1\right)}  \\
\end{align}

To simplify derivations we define 
\begin{align}
\label{eq::phi}\Phi(k, \rho_1, \rho_2) &= \frac{\left( \frac{1}{2} + \frac{1}{2}k^{\frac{1}{1+\rho_2}}\right) ^{\rho_2} \left( 1 - k^{\frac{-\rho_2}{1+\rho_2}}\right)
}{\left( \frac{1}{2} + \frac{1}{2}k^{\frac{1}{1+\rho_1}}\right) ^{\rho_1} \left( 1 - k^{\frac{-\rho_1}{1+\rho_1}}\right)}  \\
\label{eq::psi} \Psi(k, \rho_2) &= \log{\left( \frac{1}{2} + \frac{1}{2}k^{\frac{1}{1+\rho_2}}\right) } +\frac{\log{k}}{\left(1+\rho_2\right)^2} \left(-\frac{\rho_2 k^{\frac{1}{1+\rho_2}}}{1+k^{\frac{1}{1+\rho_2}}} 
+ \frac{1}{k^{\frac{\rho_2}{1+\rho_2}} - 1}\right)  \\
&= \log{\left( \frac{1}{2} + \frac{1}{2}k^{\frac{1}{1+\rho_2}}\right) } + \frac{\left( 1+k^{\frac{1}{1+\rho_2}} - \rho_2 \left( k - k^{\frac{1}{1+\rho_2}}\right)\right) \log{k}}{ (1+\rho_2)^2 \hspace{5mm} \gamma(k, \rho_2)} \nonumber
\end{align}
where
\begin{equation}\label{eq:gamma}
 \gamma(k, \rho_2) =  \left( 1+k^{\frac{1}{1+\rho_2}}\right) \left( k^{\frac{\rho_2}{1+\rho_2}} - 1\right). 
\end{equation}
Then, equation \eqref{eq::mono} equals to the product 
\begin{equation}
\frac{\partial}{\partial \rho_2} \frac{g'(\rho_2, z)}{g'(\rho_1, z)} = \Phi(h(z), \rho_1, \rho_2) \Psi(h(z), \rho_2).
\end{equation}
Let $k = h(z)\in[0, 1]$. As $k = h(z)$ is decreasing in $z$, to check the monotonicity of the above expression with respect to $z$,
we can equivalently check the monotonicity with respect to $k$ of the following expression:
\begin{equation}\label{eq::prod} 
\Phi(k, \rho_1, \rho_2) \Psi(k, \rho_2). 
\end{equation}

Taking the derivative with respect to $k$ gives
\begin{align}\label{eq:exp1}
 \frac{\partial \Phi(k, \rho_1, \rho_2) \Psi(k, \rho_2)}{\partial k} &= \Phi'(k, \rho_1, \rho_2) \Psi(k, \rho_2) + \Phi(k, \rho_1, \rho_2) \Psi'(k, \rho_2) \nonumber \\
&= \Phi(k, \rho_1, \rho_2) \Psi(k, \rho_2) \left( \frac{\partial \log{\Phi(k, \rho_1, \rho_2)}}{\partial k} + \frac{\Psi'(k, \rho_2)}{\Psi(k, \rho_2)}\right)
\end{align}
where $\Phi'(k, \rho_1, \rho_2) = \displaystyle\frac{\partial \Phi(k, \rho_1, \rho_2)}{\partial k}$, and $\Psi'(k, \rho) = \displaystyle\frac{\partial \Psi(k, \rho)}{\partial k}$.

Now, we derive the expressions in Equation \eqref{eq:exp1}:
\begin{align}
\log{\Phi(k, \rho_1, \rho_2)} &= \hspace{2mm} \rho_2 \log{\left( \frac{1}{2} + \frac{1}{2}k^{\frac{1}{1+\rho_2}}\right) } + \log{\left(  1 - k^{\frac{-\rho_2}{1+\rho_2}} \right) } \nonumber \\
&{} \hspace{2mm} - \rho_1 \log{\left( \frac{1}{2} + \frac{1}{2}k^{\frac{1}{1+\rho_1}}\right) } - \log{\left(  1 - k^{\frac{-\rho_1}{1+\rho_1}} \right) } \nonumber \\
\nonumber\\
\frac{\partial \log{\Phi(k, \rho_1, \rho_2)}}{\partial k} &= \frac{\rho_2}{1+\rho_2} \frac{k^{\frac{-\rho_2}{1+\rho_2}}}{1+k^{\frac{1}{1+\rho_2}}} + \frac{\rho_2}{1+\rho_2} \frac{k^{\frac{-\rho_2}{1+\rho_2}-1}}{1-k^{\frac{-\rho_2}{1+\rho_2}}} - \frac{\rho_1}{1+\rho_1} \frac{k^{\frac{-\rho_1}{1+\rho_1}}}{1+k^{\frac{1}{1+\rho_1}}} - \frac{\rho_1}{1+\rho_1} \frac{k^{\frac{-\rho_1}{1+\rho_1}-1}}{1-k^{\frac{-\rho_1}{1+\rho_1}}} \nonumber \\
&= \frac{\rho_2}{1+\rho_2} \frac{1+k}{k \left( 1+k^{\frac{1}{1+\rho_2}}\right)  \left(k^{\frac{\rho_2}{1+\rho_2}} -1\right)} - \frac{\rho_1}{1+\rho_1} \frac{1+k}{k \left( 1+k^{\frac{1}{1+\rho_1}}\right)  \left(k^{\frac{\rho_1}{1+\rho_1}}-1 \right)} \nonumber \\
& = F(k, \rho_2) - F(k, \rho_1) \nonumber 
\end{align}
where
\begin{equation}\label{eq:F}
F(k, \rho) = \frac{\rho}{1+\rho} \frac{1+k}{k} \frac{1}{\gamma(k, \rho)},
\end{equation}
and \\
\begin{align}
\Psi'(k, \rho_2) &= \displaystyle\frac{\partial}{\partial k} \left(\log{\left( \frac{1}{2} + \frac{1}{2}k^{\frac{1}{1+\rho_2}}\right) } +\frac{\log{k}}{\left(1+\rho_2\right)^2} \left(-\frac{\rho_2 k^{\frac{1}{1+\rho_2}}}{1+k^{\frac{1}{1+\rho_2}}} 
+ \frac{1}{k^{\frac{\rho_2}{1+\rho_2}} - 1}\right) \right)\nonumber\\
&= \hspace{3mm} \frac{k^{-\frac{\rho_2}{1+\rho_2}}}{(1+\rho_2)(1+k^{\frac{1}{1+\rho_2}})} + \frac{1}{\left( 1+\rho_2 \right)^2 k}  \left( -\rho_2 \frac{k^{\frac{1}{1+\rho_2}}}{1+k^{\frac{1}{1+\rho_2}}} + \frac{1}{k^{\frac{\rho_2}{1+\rho_2}} - 1}\right)  \nonumber \\
&\hspace{3mm} + \frac{\log{k}}{\left( 1+\rho_2 \right)^2} \left(-\frac{\rho_2 k^{-\frac{\rho_2}{1+\rho_2}}}{(1+\rho_2)\left( 1+k^{\frac{1}{1+\rho_2}}\right)^2} - \frac{\rho_2 k^{-\frac{1}{1+\rho_2}}}{(1+\rho_2) \left( k^{\frac{\rho_2}{1+\rho_2}} - 1 \right)^2} \right) \nonumber \\
&= \hspace{3mm} \frac{k+1}{\left( 1+\rho_2\right)^2  \hspace{2mm} k \hspace{2mm} \gamma(k, \rho_2)} - \frac{\rho_2 \left( k+1\right) \left( k^{\frac{\rho_2}{1+\rho_2}} + k^{\frac{1}{1+\rho_2}}\right) \log{k}}{ \left( 1+\rho_2\right)^3 \hspace{2mm} k \hspace{2mm}  \gamma^2(k,\rho_2)} \nonumber \\
\label{eq::psi_prime}&= \hspace{3mm} \frac{k+1}{\left( 1+\rho_2\right)^2  \hspace{2mm} k \hspace{2mm} \gamma(k, \rho_2)^2} \left(\gamma(k, \rho_2)  -  \left( k^{\frac{\rho_2}{1+\rho_2}} + k^{\frac{1}{1+\rho_2}}\right) \log{k^{\frac{\rho_2}{1+\rho_2}}}\right)
\end{align}
where $\gamma(k, \rho)$ is defined in Equation \eqref{eq:gamma}.

To summarize the steps so far, we have shown that the second derivative of $\tilde{f}_{\rho_1,\rho_2}(t)$ with respect to $t$ is given by
\begin{align}
 \displaystyle\frac{\partial^2 \tilde{f}_{\rho_1,\rho_2}(t)}{\partial t^2} &=  \displaystyle\frac{\partial}{\partial t} 
 \frac{\partial}{\partial \rho_2} \frac{g'(\rho_2, g^{-1}(\rho_1, t))}{g'(\rho_1, g^{-1}(\rho_1, t))} \\
&= \displaystyle\frac{\partial}{\partial z} \left(\frac{\partial}{\partial \rho_2} \frac{g'(\rho_2, z)}{g'(\rho_1, z)}\right) \displaystyle\frac{\partial z}{\partial t} \\
&=\displaystyle\frac{\partial \Phi(k, \rho_1, \rho_2) \Psi(k, \rho_2)}{\partial k}\displaystyle\frac{\partial k}{\partial z}\displaystyle\frac{\partial z}{\partial t}
\end{align}
where $z = g^{-1}(\rho_1, t)$, $k=h(z)$ with $h(z)$ defined in \eqref{eq:h}, $\Phi(k, \rho_1, \rho_2)$ given by \eqref{eq::phi}, and $\Psi(k, \rho_2)$ given by \eqref{eq::psi}.

We first prove the claims of the lemma for $\rho_1 = \rho_2 = \rho$. Coming back to Equation \eqref{eq:exp1}, 
\begin{align}
\frac{\partial \Phi(k, \rho, \rho) \Psi(k, \rho)}{\partial k}  &= \Psi'(k, \rho)
\end{align}
as $\Phi(k, \rho, \rho) = 1$, and $\displaystyle\frac{\partial \log{\Phi(k, \rho, \rho)}}{\partial k} = 0$. 
Hence to prove the convexity claims, we need to investigate the sign of $\Psi'(k, \rho)$ we derived in Equation \eqref{eq::psi_prime}. 

Note that the factor in front of the paranthesis in Equation \eqref{eq::psi_prime} is always positive for $k\in[0, 1]$, $\rho_2\in \mathbf{R}\setminus\{-1\}$, and
the term inside the paranthesis equals the function $m(k, \rho_2)$ defined in Lemma \ref{lem::sign_m} in Appendix IV.
So the sign of $\Psi'(k, \rho_2)$ is determined by the sign of $m(k, \rho_2)$. By Lemma \ref{lem::sign_m} , we have
\begin{align}
 &\Psi'(k, \rho_2)\geq 0, \quad \forall\rho_2 < -1 \\
 &\Psi'(k, \rho_2)\leq 0, \quad \forall\rho_2\in(-1, 0), \\
 &\Psi'(k, 0) = 0, \\
 &\Psi'(k, \rho_2)\leq 0, \quad \forall\rho_2\in(0, \rho^{*}(k)), \\
 &\Psi'(k, \rho^{*}(k)) = 0, \\
 &\Psi'(k, \rho_2)\geq 0, \quad \forall\rho_2\geq \rho^{*}(k).
\end{align}
where $\rho^{*}(k)\geq 3$ is a constant which depends on $k\in[0, 1]$. 

As $k$ is decreasing in $z$, which is non-increasing in $t$ when $\rho\geq 0$ by Lemma \ref{lem:g_concavity}, we have
\begin{equation}
  \displaystyle\frac{\partial^2 \tilde{f}_{\rho}(t)}{\partial t^2} 
= \underbrace{\Psi'(k, \rho)}_{\leq 0} \underbrace{\displaystyle\frac{\partial k}{\partial z}}_{< 0}\underbrace{\displaystyle\frac{\partial z}{\partial t}}_{\leq 0} \leq 0
\end{equation}
for $\rho\in[0, \rho^{*}(k)]$, and
\begin{equation}
  \displaystyle\frac{\partial^2 \tilde{f}_{\rho}(t)}{\partial t^2} 
= \underbrace{\Psi'(k, \rho)}_{\geq 0} \underbrace{\displaystyle\frac{\partial k}{\partial z}}_{< 0}\underbrace{\displaystyle\frac{\partial z}{\partial t}}_{\leq 0} \geq 0
\end{equation}
for $\rho\geq\rho^{*}(k)$. Hence, the function $\tilde{f}_{\rho}(t)$ is concave in $t$ when $\rho\in(0, 3]$ as claimed.

On the other hand, we know by Lemma \ref{lem:g_concavity} that $z$ is non-decreasing in $t$ when $\rho\in(-1, 0)$. 
Hence, the function $\tilde{f}_{\rho}(t)$ is convex in $t$ whenever $\rho\in(-1, 0)$.

Finally, when $\rho < -1$, $z$ is non-increasing in $t$  by Lemma \ref{lem:g_concavity}, so that $\tilde{f}_{\rho}(t)$ is convex in $t$. 

To prove the last claim of the lemma concerned with the case $\rho_1, \rho_2\in(0, 1]$ such that $\rho_1 < \rho_2$, we need to determine the sign of $\Psi(k, \rho_2)$.
Note that, $\Psi'(k, \rho)\leq 0$ for $\rho\in(0, 3]$ implies 
\begin{equation}
\Psi(k, \rho) \geq \displaystyle\lim_{k \to 1} \Psi(1, \rho) = \frac{2}{\left(1+\rho \right)^2}  \displaystyle\lim_{k \to 1} \frac{\log{k}}{\gamma(k, \rho)} = \frac{1}{\rho\left( 1+\rho \right) } \geq 0 \nonumber \\
\end{equation}
since
\begin{equation}
\displaystyle\lim_{k \to 1} \frac{\log{k}}{\gamma(k, \rho)} = \frac{0}{0} =  \displaystyle\lim_{k \to 1} \frac{\partial \log{k}/ \partial k}{\partial \gamma(k, \rho)/\partial k} = \displaystyle\lim_{k \to 1}  \frac{k + \rho k}{k \left(k + \rho k - k^{\frac{1}{1+\rho}} + g k^{\frac{\rho}{1+\rho}}\right) } = \frac{1+\rho}{2\rho}. \nonumber 
\end{equation}
As a result, $\Psi(k, \rho_2) \geq 0$ whenever $\rho_2\in(0, 3]$.

Recall that we are interested in the sign of the following expression
\begin{equation}
  \frac{\partial \Phi(k, \rho_1, \rho_2) \Psi(k, \rho_2)}{\partial k}  = \Phi(k, \rho_1, \rho_2) \Psi(k, \rho_2) \left(\frac{\Psi'(k, \rho_2)}{\Psi(k, \rho_2)} + F(k, \rho_2) - F(k, \rho_1)\right).  
\end{equation}
Lemma \ref{lem::F_monotonicity} in Appendix IV shows that the function $F(k, \rho)$ is decreasing in $\rho\in(0, 1]$.
Moreover, we have just shown $\displaystyle\frac{\Psi'(k, \rho_2)}{\Psi(k, \rho_2)} \leq 0$, for $\rho_2\in(0, 3]$.
Consequently, when $\rho_1, \rho_2\in(0, 1]$ such that $\rho_1 \leq \rho_2$ 
\begin{equation}
\frac{\Psi'(k, \rho_2)}{\Psi(k, \rho_2)} + F(k, \rho_2) - F(k, \rho_1) \leq 0
\end{equation}
holds, and the product $\Phi(k, \rho_1, \rho_2) \Psi(k, \rho_2)$ is non-increasing in $k$.
As $k$ is decreasing in $z$, which is in turn non-increasing in $t$ when $\rho\geq 0$ by Lemma \ref{lem:g_concavity},
the expression in equation \eqref{eq::mono}, is decreasing in $z$ whenever $\rho_1, \rho_2\in(0, 1]$ such that $\rho_1 \leq \rho_2$.
In this case, 
\begin{equation}
 \displaystyle\frac{\partial^2 \tilde{f}_{\rho_1,\rho_2}(t)}{\partial t^2} =  \underbrace{\frac{\partial \Phi(k, \rho_1, \rho_2) \Psi(k, \rho_2)}{\partial k}}_{\leq 0}
\underbrace{\displaystyle\frac{\partial k}{\partial z}}_{< 0} \underbrace{\displaystyle\frac{\partial z}{\partial t}}_{\leq 0} \leq 0,
\end{equation}
 whence the function $\tilde{f}_{\rho_1,\rho_2}(t)$ is concave in $t$ as claimed. 
\end{proof}

\subsection*{Appendix III}
\begin{proof}[Proof of Lemma \ref{lem:f_concavity}]
Taking the first derivative of $f_{\rho_1,\rho_2}(t)$ with respect to $t$, we obtain 
\begin{align}
\frac{\partial f_{\rho_1,\rho_2}(t)}{\partial t} &= \frac{\partial g(\rho_2, g^{-1}(\rho_1, t))}{\partial t}  \\
&=  \frac{g'(\rho_2, g^{-1}(\rho_1, t))}{g'(\rho_1, g^{-1}(\rho_1, t))}. \label{eq::mono}
\end{align}
Let $z = g^{-1}(\rho_1, t)$. As $g(\rho, z)$ is a monotone function in $z$ by Lemma \ref{lem:g_concavity}, so is $z = g^{-1}(\rho, t)$ in $t$. 
Hence we can check the convexity of $f_{\rho_1,\rho_2}(t)$ with respect to $t$, from the monotonicity with respect to $z$ of the following expression:
\begin{equation}
\frac{g'(\rho_2, z)}{g'(\rho_1, z)}. 
\end{equation}

Taking the derivative with respect to $z$, we get
\begin{align}
\frac{\partial}{\partial z} \frac{g'(\rho_2, z)}{g'(\rho_1, z)} &= \frac{\partial}{\partial z}  2^{\rho_1-\rho_2}  \frac{\alpha(\rho_2, z) \beta(\rho_2, z)}{\alpha(\rho_1, z) \beta(\rho_1, z)} \nonumber \\
&= \frac{ 2^{-\rho_2} \alpha(\rho_2, z) \beta(\rho_2, z)}{ 2^{-\rho_1} \alpha(\rho_1, z) \beta(\rho_1, z)} \left( \ell(\rho_2, z) -  \ell(\rho_1, z) \right)  \nonumber 
\end{align}
where
\begin{equation}
 \ell(\rho, z) = \frac{ \partial \alpha(\rho, z)/\partial z}{\alpha(\rho, z)}  + \frac{ \partial \beta(\rho, z )/\partial z}{ \beta(\rho, z)}.
\end{equation}

One can easily check that the function 
\begin{equation}
 \alpha(\rho, z) \geq 0,
\end{equation}
for any $\rho > -1$, and while the function 
\begin{equation}
\beta(\rho, z) \geq 0,
\end{equation}
for $\rho\in(-1, 0)$, we have 
\begin{equation}
\beta(\rho, z) \leq 0,
\end{equation}
for $\rho\geq 0$.

Moreover, we claim that 
\begin{equation}
\ell(\rho_2, z) -  \ell(\rho_1, z)\geq 0
\end{equation}
when $\rho_1\in(-1, 0)$, and $\rho_2 \geq 0$, or when $\rho_1\in(0, 1]$, and $\rho_2 \geq \rho_1$,
and that 
\begin{equation}
\ell(\rho_2, z) -  \ell(\rho_1, z)\leq 0
\end{equation}
when $\rho_1 > 1$, and $\rho_2\in(-1, 0)$, or when $\rho_1 > 1$, and $\rho_2\in(0, 1]$.

Therefore, if $\rho_1\in(-1, 0]$, and $\rho_2 \geq 0$, we have 
\begin{equation}
\frac{\partial}{\partial z} \frac{g'(\rho_2, z)}{g'(\rho_1, z)} = \underbrace{\frac{ 2^{-\rho_2} \alpha(\rho_2, z) \beta(\rho_2, z)}{ 2^{-\rho_1} \alpha(\rho_1, z) \beta(\rho_1, z)}}_{\leq 0} 
\underbrace{\left( \ell(\rho_2, z) - \ell(\rho_1, z) \right)}_{\geq 0} \leq 0, 
\end{equation}
and if $\rho_1\in[0, 1]$, and $\rho_2 \geq \rho_1$, we have
\begin{equation}
\frac{\partial}{\partial z} \frac{g'(\rho_2, z)}{g'(\rho_1, z)} = \underbrace{\frac{ 2^{-\rho_2} \alpha(\rho_2, z) \beta(\rho_2, z)}{ 2^{-\rho_1} \alpha(\rho_1, z) \beta(\rho_1, z)}}_{\geq 0} 
\underbrace{\left( \ell(\rho_2, z) - \ell(\rho_1, z) \right)}_{\geq 0} \geq 0. 
\end{equation}

On the other hand, if $\rho_1 > 1$, and $\rho_2\in(-1, 0)$, we have 
\begin{equation}
\frac{\partial}{\partial z} \frac{g'(\rho_2, z)}{g'(\rho_1, z)} = \underbrace{\frac{ 2^{-\rho_2} \alpha(\rho_2, z) \beta(\rho_2, z)}{ 2^{-\rho_1} \alpha(\rho_1, z) \beta(\rho_1, z)}}_{\leq 0} 
\underbrace{\left( \ell(\rho_2, z) - \ell(\rho_1, z) \right)}_{\leq 0} \geq 0, 
\end{equation}
and if $\rho_1 > 1$, and $\rho_2\in(0, 1]$, we have
\begin{equation}
\frac{\partial}{\partial z} \frac{g'(\rho_2, z)}{g'(\rho_1, z)} = \underbrace{\frac{ 2^{-\rho_2} \alpha(\rho_2, z) \beta(\rho_2, z)}{ 2^{-\rho_1} \alpha(\rho_1, z) \beta(\rho_1, z)}}_{\geq 0} 
\underbrace{\left( \ell(\rho_2, z) - \ell(\rho_1, z) \right)}_{\leq 0} \leq 0. 
\end{equation}

Recall that we are interested in the sign of the second derivative of $f_{\rho_1,\rho_2}$ with respect to $t$ given by
\begin{align}
\frac{\partial^{2} f_{\rho_1,\rho_2}(t)}{\partial t^2} &= \frac{\partial}{\partial t} \frac{\partial g(\rho_2, g^{-1}(\rho_1, t))}{\partial t}  \\
&= \frac{\partial}{\partial z}\frac{g'(\rho_2, z)}{g'(\rho_1, z)}  \frac{\partial z}{\partial t}.
\end{align}
As $z$ is non-decreasing in $t$ for $\rho_1\in(-1, 0)$, and non-increasing for $\rho_1 \geq 0$ by Lemma \ref{lem:g_concavity}, 
the function $f_{\rho_1,\rho_2}(t)$ is concave in $t$ when $\rho_1\in(-1, 0]$, and $\rho_2 \geq 0$, or when $\rho_1\in[0, 1]$, and $\rho_2 \geq \rho_1$, or when
$\rho_1 > 1$, and $\rho_2\in(-1, 0)$, and convex when $\rho_1 > 1$, and $\rho_2\in(0, 1]$.

Now, we prove the claim. For that purpose, we show that the function $\ell(\rho. z)$ is non-decreasing in $\rho$ for the interval $\rho\in(-1, 3)$, 
and $\displaystyle\frac{\partial \ell(\rho, z)}{\partial\rho}$ changes sign only once after $\rho \geq 3$. As
\begin{equation}
 \lim_{\rho\to 1} \ell(\rho, z) = \lim_{\rho\to \infty} \ell(\rho, z) = \displaystyle\frac{1}{z - z^3}
\end{equation}
holds, we conclude that
\begin{align}
 &\ell(\rho, z) \geq \ell(1, z), \quad \hbox{when  } \rho\geq 1, \\
 &\ell(\rho, z) \leq \ell(1, z), \quad \hbox{when  } \rho\in(-1, 1].
\end{align}
The above inequalities ensure $\ell(\rho_2, z) -  \ell(\rho_1, z)\geq 0$ when $\rho_1\in(-1, 0)$, and $\rho_2 \geq 0$, or when $\rho_1\in[0, 1]$, and $\rho_2 \geq \rho_1$. 
Similarly, the previous arguments ensure that $\ell(\rho_2, z) -  \ell(\rho_1, z)\leq 0$ when $\rho_1 > 1$, and $\rho_2\in(-1, 1]$.

Note that
\begin{equation}
 \ell(\rho, z) = \displaystyle\frac{\partial}{\partial z} \left(\log\left(2^{-\rho}\alpha(\rho, z)\right)  + \log\left(\beta(\rho, z)\right) \right).
\end{equation}
Hence,
\begin{align}
 \displaystyle\frac{\partial \ell(\rho, z) }{\partial \rho}&= \displaystyle\frac{\partial}{\partial z} 
\left(\frac{ \partial 2^{-\rho} \alpha(\rho, z)/\partial \rho}{2^{-\rho}\alpha(\rho, z)}  + \frac{ \partial \beta(\rho, z )/\partial \rho}{ \beta(\rho, z)} \right) \\
&= \displaystyle\frac{\partial \Psi(k, \rho)}{\partial k} \displaystyle\frac{\partial k}{\partial z} \\
&= \Psi'(k, \rho) \displaystyle\frac{\partial k}{\partial z}
\end{align}
where $k = h(z)$ is defined in Equation \eqref{eq:h}, and $\Psi'(k, \rho)$ is defined in Equation \eqref{eq::psi_prime}. 
Luckily, we have already investigated the sign of $\Psi'(k, \rho)$ in the proof of Lemma \ref{lem:ftilde} we previously stated.
Indeed, we have shown that $\Psi'(k, \rho) \leq 0$, for $\rho\in(-1, 3)$, 
and the function changes sign only once after $\rho \geq 3$. As $k$ is decreasing in $z$, the sign of $\displaystyle\frac{\partial \ell(\rho, z) }{\partial \rho}$ is exactly the opposite of $\Psi'(k, \rho)$. 
This concludes the proof.
\end{proof}

\subsection*{Appendix IV}
\begin{lemma}\label{lem::sign_m}
For $k\in[0,1]$, we define
\begin{equation}\label{eq::sign}
 m(k, \rho) = -1 + k - k^{\frac{1}{1+\rho}} + k^{\frac{\rho}{1+\rho}} - \left( k^{\frac{\rho}{1+\rho}} + k^{\frac{1}{1+\rho}}\right) \log{k^{\frac{\rho}{1+\rho}}}. 
\end{equation}            
Then, for $\forall k\in[0,1]$, we have
\begin{align*}
 &m(k, \rho)\geq 0, \quad \forall\rho < -1, \\
 &m(k, \rho) \leq 0, \quad \forall\rho\in(-1, 0), \\
 &m(k, 0) = 0. 
\end{align*}      
Moreover, $\exists \hspace{2mm} \rho^{*}(k) \geq 3$ which depends on $k$ such that:
\begin{align*}
 &m(k, \rho) \leq 0, \quad \forall\rho\in(-1, \rho^{*}(k)), \\
 &m(k, \rho^{*}) = 0, \\
 &m(k, \rho) \geq 0, \quad \forall\rho\in(\rho^{*}, \infty).
\end{align*}   
                                                                   
\end{lemma}
\begin{proof}
We now follow a series of transformations. Let
\begin{equation}
 t = \frac{\rho}{1+\rho} \nonumber
\end{equation}
Then, \eqref{eq::sign} reduces to 
\begin{equation}
 m\left(k, \frac{t}{1-t}\right) =  -1 + k - k^{1-t} + k^t - (k^t + k^{1-t}) \log{k^t}. \nonumber 
\end{equation}
In addition, let 
\begin{equation}
 s = -t \log{k}. \nonumber
\end{equation}
Then,
\begin{equation}\label{eq:H2}
 m\left(k, \frac{-s}{\log{k} + s}\right) = -1 + k - ke^{s} + e^{-s} + s(e^{-s} + ke^{s}).
\end{equation}
We first note that the function is zero at $s = 0$. Taking the first derivative with respect to $s$, we get
\begin{align*}
 \displaystyle\frac{\partial}{\partial s} m\left(k, \frac{-s}{\log{k} + s}\right) &= -ke^{s} - e^{-s} + e^{-s} + ke^{s} + s(-e^{-s} + ke^{s}) \\
&= s(-e^{-s} + ke^{s}) \\
&= t (k^{t} - k^{1-t}) \log{k}. 
\end{align*}
Hence the function $\displaystyle m\left(k, \frac{-s}{\log{k} + s}\right)$ is non-increasing in $s$ for $t\in[0, 1/2]$, and non-decreasing otherwise. 

Moreover, the derivative of $\displaystyle m\left(k, \frac{t}{1-t}\right)$ with respect to $t$ is given by
\begin{equation*}
 \displaystyle\frac{\partial}{\partial t} m\left(k, \frac{t}{1-t}\right) = \displaystyle\frac{\partial}{\partial s} m\left(k, \frac{-s}{\log{k} + s}\right)  \displaystyle\frac{\partial s}{\partial t} 
\end{equation*}
As $s$ is non-decreasing in $t$, we have shown that $\displaystyle m\left(k, \frac{t}{1-t}\right)$ is non-increasing in $t$ for $t\in[0, 1/2]$, and non-decreasing otherwise. 

Similarly, the derivative of $m(k, \rho)$ with respect to $\rho$ is given by
\begin{equation*}
 \displaystyle\frac{\partial m\left(k, \rho\right)}{\partial \rho} = \displaystyle\frac{\partial}{\partial t} m\left(k, \frac{t}{1-t}\right)  \displaystyle\frac{\partial t}{\partial \rho} 
\end{equation*}
As $t$ is increasing in $\rho$ for the intervals $(-\infty, -1)$, and $(-1, \infty)$, $m(k, \rho)$ will be non-increasing in $\rho$ for $t\in[0, 1/2]$, and non-decreasing otherwise.
We simply need to map this result to the claims of the lemma in terms of the intervals defined by $\rho$.

For the interval $t\in[1, \infty)$, we have $\rho < -1$, and $m(k, \rho)$ is non-decreasing in $\rho$.
Moreover,
\begin{equation*}
 \lim_{\rho\to-\infty} m(k, \rho) = (-1 + k -1 +k) - (k + 1) \log{k}  = -2 (1-k) + (k+1) \log{k} \geq 0
\end{equation*}
where the sign follows by noting that at $k = 1$ the expression evaluates to $0$, and it is non-increasing in $k$ as
\begin{equation*}
 \displaystyle\frac{\partial}{\partial k}\left(-2 (1-k) + (k+1) \log{k}\right)  = 1 - \displaystyle\frac{1}{k} + \log{\displaystyle\frac{1}{k}} \leq 0
\end{equation*}
using $\log{x}\leq x-1$ inequality. This shows $m(k, \rho)\geq 0$ for $\rho < -1$.

For the interval $t\in(-\infty, 0]$, we have $\rho\in(-1, 0]$, and $m(k, \rho)$ is non-decreasing in $\rho$. As we have $m(k, 0) = 0$, we conclude $m(k, \rho) \leq 0$ for $\rho\in(-1, 0)$.

For the interval $t\in[0, 1/2]$, we have $\rho\in[0, 1]$, and $m(k, \rho)$ is non-increasing in $\rho$. As we have $m(k, 0) = 0$, we conclude $m(k, \rho) \leq 0$ for $\rho\in(0, 1]$.

For the interval $t\in[1/2, 1]$, we have $\rho\geq1$, and $m(k, \rho)$ is non-decreasing in $\rho$. As $m(k, 1) \leq 0$, and
\begin{equation*}
 \lim_{\rho\to\infty} m(k, \rho) = (-1 + k -1 +k) - (k + 1) \log{k}  = -2 (1-k) + (k+1) \log{k} \geq 0,
\end{equation*}
the function will eventually cross zero. Now, we prove that the crossing point $\rho^{*}$, i.e. $m(k,  \rho^{*}) = 0$, is such that $\rho^{*}\geq 3$. For that purpose, we only need to show that $m(k, 3)$ is increasing in $k$
because $m(1, 3) = 0$ holds. 

Taking the first derivative with respect to $k$, we get
\begin{equation*}
 \displaystyle\frac{\partial m(k, 3)}{\partial k} = \displaystyle\frac{4 (-1 + k^{3/4}) - 3/4 (1 + 3 \sqrt{k}) \log{k}}{4 k^{3/4}} \geq 0
\end{equation*}
with equality iff $k = 1$. The sign follows by noting that the denominator is positive, the numerator is decreasing in $k$, and is equal to $0$ iff $k = 1$. 
Indeed, taking the first derivative with respect to $k$ of the numerator, we get
\begin{equation*}
 \displaystyle\frac{\partial}{\partial k} \left(4 (-1 + k^{3/4}) - 3/4 (1 + 3 \sqrt{k}) \log{k} \right) = \displaystyle\frac{-3 (2 + 6 \sqrt{k} - 8 k^{3/4} + 3 \sqrt{k} \log{k})}{8 k} \leq 0
\end{equation*}
with equality iff $k = 1$. The sign follows by noting that the denominator is positive, the numerator is increasing in $k$, and is equal to $0$ iff $k = 1$.
To see this, once more we take the first derivative with respect to $k$ of the numerator. Then, we get
\begin{equation*}
 \displaystyle\frac{\partial}{\partial k} \left(-3 (2 + 6 \sqrt{k} - 8 k^{3/4} + 3 \sqrt{k} \log{k})\log{k}\right)= \displaystyle\frac{-9 (4 - 4 k^{1/4} + \log{k})}{2\sqrt{k}} \geq 0
\end{equation*}
with equality iff $k = 1$. The sign follows by noting that the denominator is positive, the numerator is decreasing in $k$, and is equal to $0$ iff $k = 1$. 
To show this, we need to take the first derivative with respect to $k$ of the numerator one last time. Doing so, we get
\begin{equation*}
 \displaystyle\frac{\partial}{\partial k} \left(-9 (4 - 4 k^{1/4} + \log{k})\right) = \displaystyle\frac{9(-1 + k^{1/4})}{k} \leq 0
\end{equation*}
for $k\in[0, 1]$, and with equality iff $k = 1$. This concludes the proof of the lemma.
\end{proof}

\begin{lemma}[\cite{private:Emre_Telatar}]\label{lem::F_monotonicity}
The function $F(k, \rho)$ defined in \eqref{eq:F} is a decreasing function in $\rho\in[0, 1]$. 
\end{lemma}
\begin{proof}
For convenience, we define the function $H(k, \rho) = -\displaystyle\frac{k}{1+k} F(k, \rho)$ as 
\begin{equation}\label{eq:H}
H(k, \rho) = \frac{\rho}{1+\rho} \frac{1}{\left( 1+k^{\frac{1}{1+\rho}}\right) \left(1 - k^{\frac{\rho}{1+\rho}}\right)}  \geq 0
\end{equation}
where $k\in[0, 1]$. We note that instead of $F(k, \rho)$, we can also check the monotonicity of $H(k, \rho)$ with respect to $\rho$.\\\\
We now follow a series of transformations. Let
\begin{equation}
 t = \frac{\rho}{1+\rho} \hspace{5mm} \hbox{for} \hspace{3mm} t\in[0, \frac{1}{2}]. \nonumber
\end{equation}
Then, \eqref{eq:H} reduces to 
\begin{equation}
 H(k, \frac{t}{1-t}) = \frac{t}{\left( 1 - k^{t}\right)  \left( 1 + k^{1-t}\right)}. \nonumber 
\end{equation}
In addition, let 
\begin{equation}
 s = -t \ln{k} \hspace{5mm} \hbox{for} \hspace{3mm} s\in[0, \frac{1}{2}\ln{\frac{1}{k}}]. \nonumber
\end{equation}
Then,
\begin{equation}\label{eq:H2}
 H(k, \frac{-s}{\log{k} + s}) = \frac{1}{\log{\frac{\displaystyle 1}{\displaystyle k}}} \hspace{2mm} \frac{s}{1 - e^{-s}} \hspace{2mm} \frac{1}{1 + k e^{s}}.
\end{equation}
We note that the first fraction in \eqref{eq:H2} can be treated as a constant and we ignore it. We define the variable $a = \frac{1}{k} \geq 1$. 
For simplicity, we consider the function
\begin{equation}
\frac{1}{H(k, \frac{-s}{\log{k} + s})} = \underbrace{\frac{\ln{a}}{a}}_{constant} \frac{1 - e^{-s}}{s}\left( a + e^{s}\right). \nonumber
\end{equation}
We first show that $\ln{\left( \frac{\displaystyle 1 - e^{-s}}{\displaystyle s}\left( a + e^{s}\right)\right)}$ is a convex function for all $s\geq 0$. 
Taking the first derivative with respect to $s$, we obtain
\begin{equation}\label{eq:appEexp2}
 \frac{\partial}{\partial s} \left( -\ln{s} + \ln{\left(\frac{1}{ 1 - e^{-s}}\right)} + \ln{\left( \frac{e^{s}}{a + e^{s}}\right)}\right) = -\frac{1}{s} + \frac{e^{s}}{a + e^{s}} + \frac{1}{e^{s} - 1}. 
\end{equation}
Taking the second derivative in $s$, we get
\begin{align}
 &\frac{\partial^{2}}{\partial s^{2}} \left(-\ln{s} + \ln{\left(\frac{1}{ 1 - e^{-s}}\right)} + \ln{\left( \frac{e^{s}}{a + e^{s}}\right)} \right) \nonumber \\ 
=\hspace{2mm} &\frac{1}{s^{2}} + \frac{a e^{s}}{(a + e^{s})^{2}} - \frac{e^{s}}{(e^{s} - 1)^{2}} \nonumber \\
\geq\hspace{2mm} & \frac{1}{s^{2}} - \frac{e^{s}}{(e^{s}-1)^2} \nonumber \\
=\hspace{2mm} &\frac{1}{s^{2}} - \left( \frac{1}{e^{\frac{s}{2}} + e^{\frac{-s}{2}}}\right)^{2} \nonumber \\
=\hspace{2mm} &\frac{1}{s^{2}} - \frac{1}{\left(2 \sinh{\frac{\displaystyle s}{\displaystyle 2}}\right)^{2}} \nonumber \\
\geq\hspace{2mm} &0 \nonumber
\end{align}
where the last inequality follows from $\sinh{x}\geq x$, for $x\geq 0$.
We proved that $\ln{\left( \frac{\displaystyle 1 - e^{-s}}{\displaystyle s}\left( a + e^{s}\right)\right)}$ is a convex function for all $s\geq 0$. 
Therefore the function has only one minimum, and to decide whether the expression is decreasing in $s\in[0, \frac{1}{2}\ln{a}]$, 
it is sufficient to evaluate \eqref{eq:appEexp2} at $s = \frac{1}{2}\ln{a}$.
\begin{align}
 &\frac{\partial}{\partial s} \left(  -\ln{s} + \ln{\left(\frac{1}{ 1 - e^{-s}}\right)} + \ln{\left( \frac{e^{s}}{a + e^{s}}\right)} \right)\Bigr\rvert_{s = \frac{1}{2}\ln{a}}  \nonumber \\
= &-\frac{1}{\ln{\sqrt{a}}} + \frac{\sqrt{a}}{a + \sqrt{a}} + \frac{1}{\sqrt{a} - 1} \nonumber \\
= &-\frac{1}{\ln{\sqrt{a}}} + \frac{2\sqrt{a}}{a-1} \nonumber\\
\leq 0 \nonumber
\end{align}
since for $b = \sqrt{a} \geq 1$, we can show that
\begin{equation}\label{eq:appEexp3}
\frac{b^{2}-1}{2b} - \ln{b} \geq 0.
\end{equation}
Taking the first derivative of \eqref{eq:appEexp3} with respect to $b$, we get
\begin{equation}
 \frac{\partial}{\partial b}  \frac{b^{2}-1}{2b} - \ln{b} = \frac{1}{2} + \frac{1}{2b^{2}} - \frac{1}{b} = \frac{(b - 1)^{2}}{2 b^{2}} \geq 0. \nonumber
\end{equation}
Therefore, we proved that for each $k\in[0, 1]$ the function $\frac{1}{H(k, \frac{-s}{\log{k} + s})}$ is decreasing in $s$. 
By definition, the variable $t$ is increasing in $\rho$, and $s = -t \ln{k}$ is also increasing in $t$ for a given $k$.
As a consequence, the function $F(k, \rho) = -\displaystyle\frac{1+k}{k}H(k, \rho)$ is decreasing in $\rho$.
\end{proof}

\bibliographystyle{IEEEtran}
\bibliography{ref}
\end{document}